\documentclass[final,1p,times]{elsarticle}
\usepackage{amsmath}
\usepackage{amsthm}
\usepackage{indentfirst}
\usepackage{subfigure}
\usepackage{graphicx}
\usepackage{epstopdf}
\usepackage{epsfig}

\newtheorem{theorem}{Theorem}

\newtheorem{lemma}{Lemma}

\usepackage{algorithm, algorithmic}
\usepackage{amssymb}
\usepackage{indentfirst}
\begin{document}
\begin{frontmatter}
\title{Off-grid Multi-Source Passive Localization Using a Moving Array}
\tnotetext[t1]{This paper was funded in part by the China Postdoctoral Science Foundation, grant numbers 2017M613076 and 2016M602775; in part by the National Natural Science Foundation of China, grant numbers 61805189, 61801347, 61801344, 61522114, 61471284, 61571349, 61631019, 61871459, 61801390 and 11871392; by the Fundamental Research Funds for the Central Universities, grant numbers XJS17070, NSIY031403, and 3102017jg02014; in part by the Aeronautical Science Foundation of China, grant number 20181081003; and by the Science, Technology and Innovation Commission of Shenzhen Municipality, grant number JCYJ20170306154716846.}
\author[1]{Dan Bao}
\ead{dbao@mail.xidian.edu.cn}
\author[1]{Changlong Wang\corref{cor1}}
\ead{wangchanglong@xidian.edu.cn}
\cortext[cor1]{Corresponding author}
\author[1]{Jingjing Cai}
\ead{jjcai@mail.xidian.edu.cn}
\address[1]{ School of Electronic Engineering, Xidian University, Xi’an, Shaanxi, 710071, China }
\begin{abstract}
A novel direct passive localization technique
through a single moving array is proposed in this paper using the sparse
representation of the array covariance matrix in spatial domain. The measurement
is constructed by stacking the vectorized version of all the array covariance
matrices at different observing positions. First, an on-grid compressive sensing
(CS) based method is developed, where the dictionary is composed of the steering
vectors from the searching grids to the observing positions. Convex optimization
is applied to solve the ${\ell _1}$-norm minimization problem. Second, to get
much finer target positions, we develop an on-grid CS based method, where the
majorization-minimization technique replaces the atan-sum objective function in
each iteration by a quadratic convex function which can be easily minimized. The
objective function, atan-sum, is more similar to ${\ell _0}$-norm, and more
sparsity encouraging than the log-sum function. This method also works more
robustly at conditions of low SNR, and fewer observing positions are needed than
in the traditional ones. The simulation experiments verify the promises of the
proposed algorithm.
\end{abstract}

\begin{keyword}
Localization, array covariance matrix, off-grid
compressive sensing, multiple targets.
\end{keyword}

\end{frontmatter}

\section{Introduction}
Finding the positions of passive sources from an array of spatially
separated sensors has been of considerable interest for decades in both military
and civilian applications, such as radar, sonar, and global positioning systems,
mobile communications, multimedia, and wireless sensor networks.

There is a vast literature dedicated to the passive localization problem
applying classical signal-processing methods. The measurements needed for the
localization problem are usually the phase, strength, or time information of the
signals impinging on the antennas. Thus, the localization techniques are often
based on the time of arrival (TOA), time difference of arrival (TDOA), received
signal strength (RSS), direction of arrival (DOA) [1] [2], and phase difference
rate [3] of the intercepted signals.

The bearing only localization (BOL) or the phase difference rate based algorithms essentially make use of the phase difference information of the signals. The sensors are often modeled as a narrow bandwidth receiving antenna array. In the traditional ways, the position of the target is obtained using a two-step method. First, the DOA of the target is estimated using multiple signal classification (MUSIC) or the phase differences.Second, the position of the target is estimated using those multiple measurements of the DOA or the phase difference rates. However, this two-step localization approach suffers from the nonlinear relationships between the phase difference and the target position. It also has poor signal-to-noise ratio (SNR) performance.

Furthermore, the localization researches usually simply deal with the
single-target problems traditionally. The multiple-target- localization problem
should be divided into multiple single-target-localization problems only if the
measurement can be uniquely assigned to individual targets. A one-step
localization for multiple sources are proposed in [4] by communicating the
estimated covariance matrix between decentralized sub-arrays, which processing
only increases the load of the communication slightly. This is a MUSIC like
algorithm based on subspace decomposition. It has the ability of positioning
multiple targets simultaneously. Same idea is used in [5] for a moving array
localization problem, where a Cram\'{e}r-Rao Bound (CRB) is given. Along with the
development of compressive sensing (CS), a joint sparse representation of array
covariance matrices (JSRACM) for emitter locations on multiple phase arrays is
proposed in [6]. The locations may be estimated by solving an unconstrained
optimization problem.
\subsection{Compressed sensing and Location problems}
Compressed sensing which is also called sparse recovery has witnessed an increasing interest in recent years to meet
the high demand for efficient information acquisition scheme [7]. Contrary to the
traditional Nyquist criteria, CS, depending on finding sparse solutions to
underdetermined linear systems, can reconstruct the signals from far fewer
samples than is possible using Nyquist sampling rate. CS has seen major
applications in diverse fields, ranging from image processing to array signal
processing. These applications have been successful because of the inherent
sparsity of many real-world signals like sound, image, and video.

The early works in CS assume the sparse solutions lie on some fixed grids.
However, this is not true in practical applications, so off-grid CS is proposed.
In off-grid scenarios, an atomic norm minimization approach is proposed in [8] to
exactly recover the unknown waveform. The problem is solved by reformulating it
as an exact semidefinite program. In [9], atomic norms are further studied in the
line spectral estimation in a noisy condition. Alternatively, an iterative
reweighted method for off-grid CS is proposed in [10], where the sparse signals
and the unknown parameters used to construct the true dictionary are jointly
estimated.

As soon as the concept of the off-grid CS emerged, it was quickly applied to DOA
and localization problems. In [11], joint sparsity reconstruction methods is used
to estimate the DOA by exploring the underlying structure between the sparse
signal and the gird mismatch for off-grid targets. The idea is to approximate the
measurement matrix by using the first order Taylor expansion around the
predefined grid. Atomic norm based off-grid CS is used to the same DOA problem in
[12], which achieves high-resolution DOA estimation through the polynomial
rooting method. The counting and localization problem for off-grid targets in
wireless sensor networks is also formulated to a sparse recovery problem [13],
where the true and unknown sparsifying dictionary is approximated with its first
order Taylor expansion around a known dictionary. To locate multiple sources
through TDOA measurements, [14] proposes a new Bayesian learning method for cases
where the off-grid error is considered.
\subsection{The main contribution of this paper}
In this paper, we propose a novel localization algorithm and its theoretical analysis which is designed for multi-source passive localization. To summarize, the main contribution of this paper is as follows,

\romannumeral1. By the array covariance matrix based on off-grid CS, a novel sparse recovery algorithm which has the ability of
estimating multiple targets positions simultaneously. A better localization performance can be achieved by fewer antenna elements and fewer observing positions compared to the traditional BOL methods.

\romannumeral2. Due to the difficult caused by $\ell_0$-minimization, we use arc-tangent function to replace the original $0$-norm. In order to ensure that both the NP-HARD original model and the continuous alternative optimization model share the same sparse solution, we prove the equivalence between these two models which can explain the reason theoretically why the proposed method performs better than classic methods.

In the proposed method, the measurement of the localization formulation is obtained by stacking all the
estimated array covariance matrices from different observing positions together.
To begin with, an on-grid localization is achieved simply by sparsely
representing the intercepted signals in the spatial domain with an over-complete
basis. Convex optimization is then used to solve the constrained ${\ell _1}$-norm
minimization problem.

However, to get a more precise target position, much finer grids are needed,
such that the size of the dictionary will increase dramatically, especially in a
three dimensional positioning problem. With the grids getting more and more fine,
this approach may suffer from a prohibitively high computational complexity, and
the coherence between the atoms of the dictionary increases. Enlightened by
off-grid CS in [10] and the applications, an iterative reweighted algorithm with
majorization-minimization (MM) is introduced to the covariance matrixes based
direct localization method. The MM technique replaces the atan-sum objective
function in each iteration by a quadratic convex function which can be minimized
easily. The objective function, atan-sum, is more similar to ${\ell _0}$-norm,
and more sparsity encouraging than the log-sum function. Therefore, the main contribution of this paper is to prove the equivalence relationship between ${\ell _0}$-norm and atan-sum objective function.

The remaining parts of the paper are organized as follows. In Section II, we
begin with the definition of the system model. The on-grid sparsified
localization model is given, and convex optimization is used to solve the
problem. Section III proposes off-grid compressive sensing based localization.
Section IV provides some simulation examples to verify the performance of the
proposed algorithm. Finally, Section V concludes the paper. In the Appendix, the
derivation of the cost function with respect to the coordinate is calculated, and
the CRB of the localization problem is given.

\begin{figure}
\centering
\includegraphics[width=0.7\textwidth]{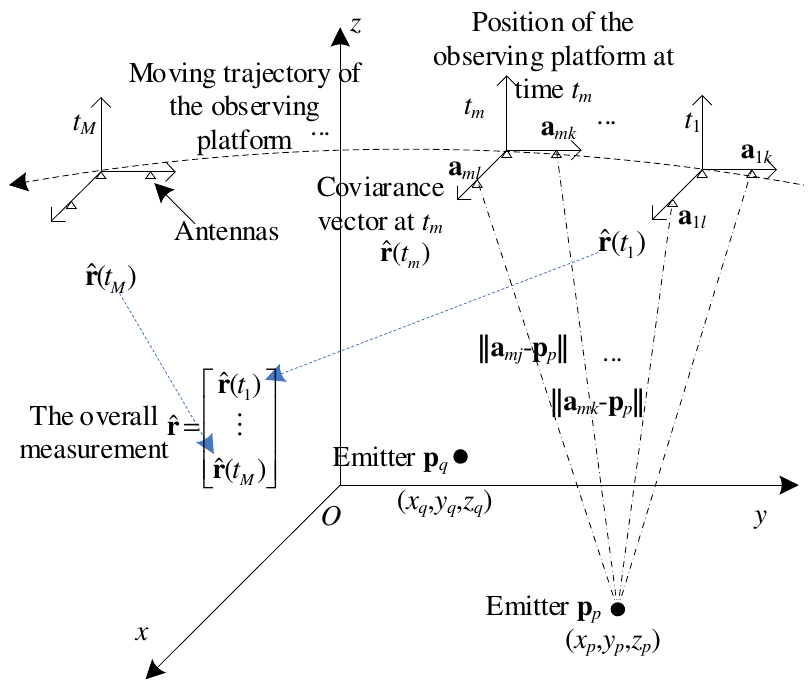}
\caption{The localization model.}
\label{1}
\end{figure}
\subsection{Notations}
In this paper, diag( ) denotes making a diagonal matrix with the elements of the
vector. The notations $|| \cdot ||{}_1$, $|| \cdot ||$ denote the ${\ell _1}$,
and ${\ell _2}$-norm of the matrix, respectively. The superscripts $^{\rm{*}}$,
$^T$, $^H$ denote the conjugation, transposition, conjugate transposition of the
matrix, respectively. The notation vec( ) denotes the vectorization operator of a
matrix.

\section{Covariance sparse representation based multi-target localization}
We look into a localization problem for multiple targets from a moving platform.
As an example shown in Fig. 1, an observing platform, such as a satellite or an
unmanned aerial vehicle, moves along its trajectory. A planar antenna array is
installed on the moving platform to intercept the signals transmitted from the
emitters. The real-time positions of the platform can be obtained by its own
positioning devices, such as a GPS receiver. The task of the localization problem
is to estimate the positions of the targets.

In the traditional ways, the positions of the targets are obtained using a BOL
or the phase difference rate. Here, we propose a new signal model to solve the
localization problem.

Without loss of generality, the coordinate system of the moving platform is the
same with geodetic coordinate system all the time for simplicity. The signals
could not be intercepted all the time, but only are available at discontinuous
slow time instants ${t_m}$, $m = 1,...,M$, because of the noncooperative
receiver. We do not require that the signals are coherent for all the slow time
instants because the crystal on the receiver of the moving platform could not
keep stable for a long period. Instead, we just need to calculate the covariant
matrix of the signals from the antenna array only at the instant ${t_m}$, $m =
1,...,M$, respectively.

The planar array consists of \textit{L} antenna elements. Let ${{\bf{a}}_{ml}}$
denote the position of the antenna \textit{l}, \textit{l}=1,\ldots{},\textit{L}
at the instant ${t_m}$, $m = 1,...,M$. The antenna array is either uniformly or
randomly distributed, while the latter can break through the half-wavelength
limitation of the array, eliminate the grating lobe, and result in improved
property for the over-complete dictionary and offer more robust reconstruction.
The receiver intercepts the signals emitted from the target position
${{\bf{p}}_p}: = {[{x_p},{y_p},{z_p}]^T}$, $p = 1,...,K$, where \textit{K} is the
number of the emitters. The task of the localization is to estimate
${{\bf{p}}_p}$ from the signals impinging on the array at several instants
${t_m}$, $m = 1,...,M$. Here, the time instants ${t_m}$ are called the slow time.

The \textit{K} narrowband signals, ${{\bf{s}}_{\rm{r}}}({t_{m,n}}) \in {^L}$,
impinging on the \textit{L} antennas are received during a short period at time
${t_m}$ containing \textit{N} snapshots, which can be expressed as:
\begin{equation}
{{\bf{s}}_{\rm{r}}}({t_{m,n}}) = {\bf{A}}({t_m}){\bf{s}}({t_{m,n}}) +
{{\bf{n}}_{\rm{r}}}({t_{m,n}})
\label{eq:1}
\end{equation}
where ${\bf{s}}({t_{m,n}}) \in {^{K \times 1}}$ is complex envelope of the
incoherent signals transmitted from \textit{K} emitters, ${t_{m,n}},n = 1,...,N$
is the fast time\textit{ }instant. Assume that the interval between the slow time
${t_m}$ and ${t_{m + 1}}$ is much greater than the period of \textit{N} snapshots
lasting in the fast time. The entries of ${{\bf{s}}_{\rm{r}}}({t_{m,n}})$ are the
complex envelopes of the summation of all the signals received on the antennas at
the time instant ${t_{m,n}}$. The noise vector ${{\bf{n}}_{\rm{r}}}({t_{m,n}})$
is the additive circular complex Gaussian white noise with zero mean and variance
$\sigma _n^2$, which is uncorrelated with ${\bf{s}}({t_{m,n}})$. The notation
${\bf{A}}({t_m}) \in {{\bf{C}}^{K \times L}}$ is the array-steering matrix at
time ${t_m}$, and it may be varying with time ${t_m}$ because of the moving
measurement platform. However, we assume that ${\bf{A}}({t_m})$ is constant
within the period of \textit{N} snapshots lasting in the fast time, which is
defined as

\begin{equation}
{\bf{{\rm A}}}({t_m}) = [{{\bf{\alpha }}_1}({t_m}),{\kern 1pt} {\kern 1pt}
{{\bf{\alpha }}_2}({t_m}){\kern 1pt} {\kern 1pt} ,{\kern 1pt} {\kern 1pt}
...{\kern 1pt} {\kern 1pt} {\kern 1pt} {\kern 1pt} ,{{\bf{\alpha }}_K}({t_m})]
\label{eq:2}
\end{equation}
where
\[
{{\bf{\alpha }}_p}({t_m}) = {[1{\kern 1pt} {\kern 1pt} {\kern 1pt} {\kern 1pt}
{e^{ - j2\pi f(||{{\bf{a}}_{m2}} - {{\bf{p}}_p}|| - ||{{\bf{a}}_{m1}} -
{{\bf{p}}_p}||)/c}}...{e^{ - j2\pi f(||{{\bf{a}}_{mL}} - {{\bf{p}}_p}|| -
||{{\bf{a}}_{m1}} - {{\bf{p}}_p}||)/c}}]^T},
\]
is the array-steering vector of frequency $f$ and is with respect to the
emitter's position ${{\bf{p}}_p}$. As illustrated in Fig. 1, ${{\bf{a}}_{ml}}$ is
the position of the \textit{l}th antenna at the time instant ${t_m}$, and
$||{{\bf{a}}_{ml}} - {{\bf{p}}_p}||$ is the distance between the \textit{l}th
antenna and the \textit{p}th emitter. Different from the traditional approaches
in [5] and [6], the array-steering matrix here is more directly and conveniently
defined with respect to the emitter's position ${{\bf{p}}_p}$, without any
approximation. This definition is easier to calculate the derivations, which will
be given in the Appendix.

Thus, the array covariance matrix of the received signals at time ${t_m}$ is
given by
\begin{equation}
{\bf{R}}({t_m}) = E[{{\bf{s}}_{\rm{r}}}({t_m}){\bf{s}}_{\rm{r}}^H({t_m})] =
{\bf{A}}({t_m}){\bf{\mathcal{S} }}{{\bf{A}}^H}({t_m}) + \sigma _n^2{{\bf{I}}_L}
\label{eq:3}
\end{equation}
where E[ ] denotes the expectation operator, and ${{\bf{I}}_L}$is the identity
matrix of order $L \times L$. The source covariance matrix ${\bf{\mathcal{S} }}$ is
diagonal defined as ${\bf{\mathcal{S} }} = E[{\bf{s}}({t_m}){{\bf{s}}^H}({t_m})] =
$$diag({[\sigma _1^2,...,\sigma _K^2]^T})$, where $\sigma _1^2,...,\sigma _K^2$
denote the power of the signals.

Let ${\bf{r}}({t_m})$ be the covariance vector defined as the vectorized
covariance matrix ${\bf{r}}({t_m}) = vec({\bf{R}}({t_m}))$. From \eqref{eq:3}, we
get [15]

\begin{equation}
{\bf{r}}({t_m}) = {\bf{\Psi }}({t_m}){\boldsymbol{\rho}} + \sigma _n^2vec({{\bf{I}}_L})
\label{eq:4}
\end{equation}
where
\begin{equation}
{\bf{\Psi }}({t_m}): = [{\bf{\alpha }}_{{{\bf{p}}_1}}^*({t_m}) \otimes
{{\bf{\alpha }}_{{{\bf{p}}_1}}}({t_m}),...,{\bf{\alpha }}_{{{\bf{p}}_K}}^*({t_m})
\otimes {{\bf{\alpha }}_{{{\bf{p}}_K}}}({t_m})]
\label{eq:5}
\end{equation}
${\boldsymbol{\rho}}: = {[\sigma _1^2,...,\sigma _K^2]^T}$, and $ \otimes $ denotes the
Kronecker product.

Equation \eqref{eq:4} is just a theoretical model of the array covariance vector.
However, in reality, ${\bf{r}}({t_m})$ must be estimated by

\begin{equation}
{\bf{\hat r}}({t_m}) = vec[{\bf{\hat R}}({t_m})] = vec\left[
{\frac{1}{N}\sum\limits_{n = {\rm{1}}}^N
{{{\bf{s}}_{\rm{r}}}({t_{m,n}}){\bf{s}}_{\rm{r}}^H({t_{m,n}})} } \right]
\label{eq:6}
\end{equation}
Therefore, ${\bf{\hat r}}({t_m})$ is not an ideal covariance vector any more,
while it contains some noise terms, such as the signal-to-signal and the
signal-to-noise cross-terms. Thus, the model can be rewritten as
\begin{equation}
{\bf{\hat r}}({t_m}) = {\bf{\Psi }}({t_m}){\boldsymbol{\rho}} + \sigma
_n^2vec({{\bf{I}}_L}) + {\bf{n}}({t_m})
\label{eq:7}
\end{equation}
where ${\bf{n}}({t_m})$ is the estimation error following the asymptotically
Gaussian distribution [16]

\begin{equation}
{\bf{n}}({t_m})\sim(0,\frac{1}{N}{{\bf{R}}^T}({t_m}) \otimes {\bf{R}}({t_m}))
\label{eq:8}
\end{equation}

Obviously, ${\boldsymbol{\rho}}$ cannot be solved from a measurement of a single
observation position. Therefore, we use all the measurements of different time
instants from different observing positions. Thus, we stack all the measurements
at all the available instants to form a complete model
\begin{equation}
{\bf{\hat r}} = {\bf{\Psi \rho }} + {\bf{n}}
\label{eq:9}
\end{equation}
where
\[
{\bf{n}} = {[\begin{array}{*{20}{c}}
{{{\bf{n}}^T}({t_1})}& \cdots &{{{\bf{n}}^T}({t_M})}
\end{array}]^T},
\]

\[
{\bf{\hat r}} = \left[ {\begin{array}{*{20}{c}}
{{\bf{\hat r}}({t_1})}\\
\cdots \\
{{\bf{\hat r}}({t_M})}
\end{array}} \right] - \sigma _n^2\left[ {\begin{array}{*{20}{c}}
{vec({{\bf{I}}_L})}\\
\cdots \\
{vec({{\bf{I}}_L})}
\end{array}} \right],
\]

\begin{equation}
{\bf{\Psi }} = \left[ {\begin{array}{*{20}{c}}
{{\bf{\Psi }}({t_1})}\\
\vdots \\
{{\bf{\Psi }}({t_M})}
\end{array}} \right]
\label{eq:10}
\end{equation}
The construction of the measurement ${\bf{\hat r}}$ is also illustrated as in
Fig. 1. In the measuring model of \eqref{eq:9} for multiple instants, the estimation
error of the covariance vector also follows the asymptotically Gaussian
distribution, that is,
\begin{equation}
{\bf{n}}\sim(0,{\bf{C}})
\label{eq:11}
\end{equation}
with the covariance matrix
\begin{equation}
{\bf{C}} = \frac{1}{N}\left[ {\begin{array}{*{20}{c}}
{{{\bf{R}}^T}({t_1}) \otimes {\bf{R}}({t_1})}&{}&{}\\
{}& \ddots &{}\\
{}&{}&{{{\bf{R}}^T}({t_M}) \otimes {\bf{R}}({t_M})}
\end{array}} \right]
\label{eq:12}
\end{equation}

\section{Off-grid Compressive Sensing Based Localization}
\subsection{Off-grid Localization Model}
In the classical algorithms, the positions of candidate emitters are assumed to lie
on some fixed discrete grids, so that the conventional compressive sensing can
take into effect. The continuous space has to be discretized to a finite set of
points, and signals are represented sparsely in a fixed known dictionary.
However, in practical localization problems, the emitters do not necessarily lie
on these grids exactly. For the purpose of obtaining a higher estimation
precision, the grids are needed to be refined. However, the refined grids will
increase the coherence of the measurement matrix, which
will in turn limit the performance of the algortihm [8].

Thus, the dictionary should be represented by the positions in a continuous domain. The dictionary can’t be preset any more. The off-grid localization becomes a joint dictionary learning and emitter position estimation problem.

Without dividing the target area into uniform grids, the signal model of \eqref{eq:9} should
be modified to:

\begin{equation}
{\bf{\hat r}} = {\bf{\Psi }}(\mathcal{S}){\boldsymbol{\rho}} + {\bf{n}}
\label{eq:17}
\end{equation}
where ${\bf{\Psi }}(\mathcal{S}): = [{\boldsymbol{\psi }}({{\bf{p}}_1}),...,{\boldsymbol{\psi
}}({{\bf{p}}_K})] \in {^{{L^2}M \times P}}$, ${\bf{\hat r}}$, ${\boldsymbol{\rho}}$ and
${\bf{n}}$ are the same as defined in \eqref{eq:9}. Define a set of the emitters'
positions, $\mathcal{S}: = \{ {{\bf{p}}_p}{\rm{\} }}_{p = 1}^P$, and ${\boldsymbol{\psi
}}({{\bf{p}}_p})$, the column of ${\bf{\Psi }}(\mathcal{S})$, is the function of a position
${{\bf{p}}_p}$ in the continuous domain.

Thus, the off-grid localization problem is not only estimating the signal power,
but also adjusting the position parameters of the atoms, ${\boldsymbol{\psi
}}({{\bf{p}}_p})$, to attain their true values. Off-grid CS offers improved
resolution and robust performance because of the sparsity constraint. The goal is
to search for the unknown dictionary composed of as few atoms as possible.
Thus, the off-grid localization problem can expressed as the following $l_0^{\lambda}$-minimization:

\begin{equation}
\mathop {\min }\limits_{\mathcal{S},{\boldsymbol{\rho}}} \quad
\lambda ||\boldsymbol{\rho}||_0
+ ||{\bf{\hat r}} - {\bf{\Psi }}(\mathcal{S}){\boldsymbol{\rho}}|{|^2}
\label{eq:18}
\end{equation}
\subsection{Sparse recovery model and its alternative model}
Because $\ell_0$-norm is a discrete integer function, $\ell_0^{\lambda}$-minimization is a NP-HARD problem \cite{x1}, so we consider the following alternative function $g_\delta(\cdot)$ to replace $\ell_0$-norm:
\begin{equation}
\mathop {\min }\limits_{\mathcal{S},{\boldsymbol{\rho}}} \quad \mathcal{G}(\mathcal{S},{\boldsymbol{\rho}}): = \lambda g({\boldsymbol{\rho}})
+ ||{\bf{\hat r}} - {\bf{\Psi }}(\mathcal{S}){\boldsymbol{\rho}}|{|^2},
\label{eq:19}
\end{equation}
The objective function 
\begin{eqnarray}
g_\delta({\boldsymbol{\rho}}) = \sum\limits_{p = 1}^P {g_c(\rho_p)} = \sum\limits_{p = 1}^P {atan|\rho_p/\delta |}
\end{eqnarray}
can be separated to a summation of several objective function, where ${\rho _p}$ represents the
\textit{p}th entry of ${\boldsymbol{\rho}}$. The scalar function ${atan|\rho
/\delta |}$ is sign
invariant and concave-and-monotonically increasing on the non-negative orthant
${\mathcal{O}_1}$. $\lambda  > 0$ is a regularization parameter to provide a trade-off between
fidelity to the measurements and sparsity in the optimal solution.

However, the minimization problem \eqref{eq:18} and \eqref{eq:19} are difficult to solve since there are two optimization variables $\mathcal{S}$ and $\boldsymbol{\rho}$, so we adopt alternating optimization strategy to solve this problem. Among the processing of alternating optimization, the update of $\boldsymbol{\rho}$ is the most important aspect, since a accurate solution of the present sparse model can provide a more reasonable modification of the grid point set $\mathcal{S}$. 

In sparse recovery theory, the main algorithms designed for solve model (\ref{eq:18}) can be divided into two categories, greedy algorithms, such as OMP and convex relaxation methods, such as $l_1$-minimization. Although greedy algorithms are designed to solve model (\ref{eq:18}) directly, due to the fact that model (\ref{eq:18}) is a NP-HARD problem \cite{x1}, these algorithms only performance well with a low level. Furthermore, convex relaxation methods need the measurement matrix ${\bf{\Psi }}(\mathcal{S}$ to meet the Restricted Isometry Property (RIP). A matrix $A$ is said to satisfy RIP of order $2k$ if and only if there exists a constant $\delta _{2k}\in (0,1)$ such that
\begin{eqnarray}
(1-\delta_{2k})\|x\|_2^2\leq \|Ax\|_2^2 \leq (1+\delta_{2k})\|x\|_2^2
\end{eqnarray}
for any $2k$ sparse vector $x$. In \cite{x3}, it has been proved that $\delta_{2k}\leq \sqrt{2}/2$ is the theoretical optimal conditions to recover the real sparse solution by $l_1$-minimization. However, to verify RIP for a given matrix ${\bf{\Psi }}(\mathcal{S}$ also is a NP-HARD problem \cite{x4}. Therefore, at each iteration, it is important to ensure $g_\delta({\boldsymbol{\rho}})$ can recover the real sparse solution without any request of ${\bf{\Psi }}(\mathcal{S})$.
Furthermore, the following theorem offer us a guarantee for the alternative function $g_\delta({\boldsymbol{\rho}})$ for sparse recovery.
\begin{theorem}\label{theorem1}
For a given $\mathcal{S}$, there exist a constant $\delta(\mathcal{S},\lambda)$ such that both of the following minimization problems 
\begin{equation}
\mathop {\min }\limits_{{\boldsymbol{\rho}}} \quad
\lambda ||\boldsymbol{\rho}||_0
+ ||{\bf{\hat r}} - {\bf{\Psi }}(\mathcal{S}){\boldsymbol{\rho}}|{|^2}
\end{equation}
and
\begin{equation}
\mathop {\min }\limits_{{\boldsymbol{\rho}}} \quad \mathcal{G}(\mathcal{S},{\boldsymbol{\rho}}): = \lambda g({\boldsymbol{\rho}})
+ ||{\bf{\hat r}} - {\bf{\Psi }}(\mathcal{S}){\boldsymbol{\rho}}|{|^2}
\end{equation}
share the same sparse solutions whenever $0<\delta<\delta(\mathcal{S},\lambda)$.
\end{theorem}
\begin{proof}
See Appendix for more details.
\end{proof}

\subsection{Iterative Reweighted MM Algorithm }
However, the complicated objective functions render the direct problem solving
intractable. A simple iterative algorithm is proposed here by exploiting the MM
technique to minimize the objective function of \eqref{eq:19} by using a
surrogate convex function. In each iteration, the algorithm alternates between
estimating ${\boldsymbol{\rho}}$ and refining the location set ${\mathcal{S}}$.

The goal of the proposed algorithm is to find a convex smooth function
${f_{\rm{c}}}(\rho )$ to replace ${g_{\rm{c}}}(\rho )$ in an iterative way, and
${f_{\rm{c}}}(\rho )$ should be chosen for easy minimizing. According to the MM
technique, the surrogate function should be selected satisfying two conditions.
One is it should majorize the original function, the other is its global
minimization must guarantee the original objective function reducing.

Many methods are proposed for constructing surrogate objective functions. One of them is construction by second order Taylor expansion [21]. There exists a value   such that  , making the following inequality hold [20]:

\begin{equation}
\begin{array}{l}
{g_{\rm{c}}}(\rho ) \le \\
{g_{\rm{c}}}({\rho ^{(i)}}) + {\left. {\frac{{\partial {g_{\rm{c}}}(\rho
)}}{{\partial \rho }}} \right|_{\rho  = {\rho ^{(i)}}}}(\rho  - {\rho ^{(i)}}) +
\frac{1}{2}{\beta ^{(i)}}{(\rho  - {\rho ^{(i)}})^2}
\end{array}
\label{eq:20}
\end{equation}
where ${\partial {g_{\rm{c}}}(\rho ) \over
{\partial \rho }} = {\delta\rho\over {|\rho|(\delta^2+\rho^2)}}$ is the gradient of ${g_{\rm{c}}}({\rho})$. Hence, we can define a surrogate function for the \textit{i}th iteration:

\begin{equation}
\begin{array}{l}
{f_{\rm{c}}}(\rho |{\rho ^{(i)}}) = \\
{g_{\rm{c}}}({\rho ^{(i)}}) + {\left. {\frac{{\partial {g_{\rm{c}}}(\rho
)}}{{\partial \rho }}} \right|_{\rho  = {\rho ^{(i)}}}}(\rho  - {\rho ^{(i)}}) +
\frac{1}{2}{\beta ^{(i)}}{(\rho  - {\rho ^{(i)}})^2}
\end{array}
\label{eq:21}
\end{equation}

From \eqref{eq:18}, we notice that $g({\boldsymbol{\rho}})$ is a summation of
${g_{\rm{c}}}(\rho )$, so the surrogate function of $g({\boldsymbol{\rho}})$ for the
\textit{i}th iteration, denoted by $f({\boldsymbol{\rho}}|{{\boldsymbol{\rho}}^{(i)}})$, can be
derived from the above scalar component ${f_{\rm{c}}}(\rho |{\rho ^{(i)}})$:

\begin{equation}
\begin{array}{l}
f({\boldsymbol{\rho}}|{{\boldsymbol{\rho}}^{(i)}}) = \sum\limits_{p = 1}^P {{f_{\rm{c}}}({\rho
_p}|\rho _p^{(i)})} \\
= g({{\boldsymbol{\rho}}^{(i)}}) + {({\boldsymbol{\rho}} - {{\boldsymbol{\rho}}^{(i)}})^T}{\left.
{\frac{{\partial g({\boldsymbol{\rho}})}}{{\partial {\boldsymbol{\rho}}}}} \right|_{{\bf{\rho
}} = {{\boldsymbol{\rho}}^{(i)}}}} + {({\boldsymbol{\rho}} - {{\bf{\rho
}}^{(i)}})^T}{{\bf{B}}^{(i)}}({\boldsymbol{\rho}} - {{\boldsymbol{\rho}}^{(i)}})
\end{array}
\label{eq:22}
\end{equation}
where ${\left. {{\textstyle{{\partial g({\boldsymbol{\rho}})} \over {\partial {\bf{\rho
}}}}}} \right|_{{\boldsymbol{\rho}} = {{\boldsymbol{\rho}}^{(i)}}}} = {[{\left.
{{\textstyle{{\partial {g_{\rm{c}}}(\rho )} \over {\partial \rho }}}}
\right|_{\rho  = \rho _1^{(i)}}},...,{\left. {{\textstyle{{\partial
{g_{\rm{c}}}(\rho )} \over {\partial \rho }}}} \right|_{\rho  = \rho
_P^{(i)}}}]^T}$ is the gradient of $g({\boldsymbol{\rho}})$ at ${\boldsymbol{\rho}} =
{{\boldsymbol{\rho}}^{(i)}}$, and

\begin{equation}
{{\bf{B}}^{(i)}} = diag({\textstyle{{\rm{1}} \over {\rm{2}}}}\beta
_1^{(i)},...,{\textstyle{{\rm{1}} \over {\rm{2}}}}\beta _P^{(i)})
\label{eq:23}
\end{equation}
at the \textit{i}th iteration. Ignoring terms irrelevant to the variable
${\boldsymbol{\rho}}$, the surrogate function becomes to the following equivalent form:

\begin{equation}
f({\boldsymbol{\rho}}|{{\boldsymbol{\rho}}^{(i)}}) = {{\boldsymbol{\rho}}^T}\left( {{{\left.
{\frac{{\partial g({\boldsymbol{\rho}})}}{{\partial {\boldsymbol{\rho}}}}} \right|}_{{\bf{\rho
}} = {{\boldsymbol{\rho}}^{(i)}}}} - 2{{\bf{B}}^{(i)}}g({{\boldsymbol{\rho}}^{(i)}})} \right) +
{{\boldsymbol{\rho}}^T}{{\bf{B}}^{(i)}}{\boldsymbol{\rho}}
\label{eq:24}
\end{equation}
Inspired by the idea presented in [20], \eqref{eq:24} can be further simplified to
a more compact version by selecting ${{\bf{B}}^{(i)}}$ to make ${\left.
{{\textstyle{{\partial g({\boldsymbol{\rho}})} \over {\partial {\boldsymbol{\rho}}}}}}
\right|_{{\boldsymbol{\rho}} = {{\boldsymbol{\rho}}^{(i)}}}} - 2{{\bf{B}}^{(i)}}g({{\bf{\rho
}}^{(i)}}) = {\rm{0}}$. Finally, the surrogate function can be expressed as:

\begin{equation}
f({\boldsymbol{\rho}}|{{\boldsymbol{\rho}}^{(i)}}) = {{\boldsymbol{\rho}}^T}{{\bf{B}}^{(i)}}{\boldsymbol{\rho
}}
\label{eq:25}
\end{equation}
where the diagonal elements satisfy

\begin{equation}
\beta _p^{(i)} = {{{{\left. {{\textstyle{{\partial {g_{\rm{c}}}(\rho )} \over
{\partial \rho }}}} \right|}_{\rho  = \rho _p^{(i)}}}} \mathord{\left/
{\vphantom {{{{\left. {{\textstyle{{\partial {g_{\rm{c}}}(\rho )} \over
{\partial \rho }}}} \right|}_{\rho  = \rho _p^{(i)}}}} {{g_{\rm{c}}}(\rho
_p^{(i)})}}} \right.
\kern-\nulldelimiterspace} {\rho _p^{(i)}}} = {\delta\over {|\rho|(\delta^2+\rho^2)}}
\label{eq:26}
\end{equation}


Resolving the unconstrained optimization of (19) by an iterative optimization method using the surrogate function of (25), the sub-problem of (19) at the (\textit{i})th iteration can be expressed as:

\begin{equation}
\mathop {\min }\limits_{\mathcal{S},{\boldsymbol{\rho}}}\mathcal{L}(\mathcal{S},{\boldsymbol{\rho}}): = \lambda {{\boldsymbol{\rho}}^T}{{\bf{B}}^{(i)}}{\boldsymbol{\rho}} + ||{\bf{\hat r}} - {\bf{\Psi
}}(\mathcal{S}){\boldsymbol{\rho}}|{|^2}
\label{eq:27}
\end{equation}

\subsection{Alternatively Optimizing }

From \eqref{eq:27}, it is still intractable to optimize the objective
function respect to  and ${\boldsymbol{\rho}}$ simultaneously. Hence,  and ${\boldsymbol{\rho}}$ should be optimized alternatively at the (\textit{i}+1)th iteration. First,
suppose  is estimated at the \textit{i}th iteration, so conditioned on , taking
the complex gradient [22] of it with respect to ${\boldsymbol{\rho}}$, and setting the
gradient to zero and solving the equation, one can obtain the minimum point at
the (\textit{i}+1)th iteration as follows:

\begin{equation}
\begin{array}{c}
{{\boldsymbol{\rho}}^{(i + 1)}} = {\left( {{\mathop{\rm Re}\nolimits} \left(
{{{\bf{\Psi }}^H}({\mathcal{S}^{(i)}}){\bf{\Psi }}(\mathcal{S}{^{(i)}})} \right) +
{\lambda }{{\bf{B}}^{(i)}}} \right)^{ - 1}}\\
\cdot {\mathop{\rm Re}\nolimits} \left( {{{\bf{\Psi
}}^H}(\mathcal{S}{^{(i)}}){\bf{\hat r}}} \right)
\end{array}
\label{eq:28}
\end{equation}
We observe that the optimal ${{\boldsymbol{\rho}}^{(i + 1)}}$ can be expressed by an
explicit solution such that the calculation complexity of this step can be
dramatically reduced compared with ${\ell _1}$ minimization method such as
described in [19].

\begin{theorem}\label{theorem2}
Given a solution ${{\boldsymbol{\rho}}^{(i)}}$ of the \textit{i}th
iteration, ${{\bf{B}}^{(i)}}$ is calculated by using \eqref{eq:26}. If a new
solution of the (\textit{i}+1)th iteration, ${{\boldsymbol{\rho}}^{(i + 1)}}$, is
obtained by \eqref{eq:28}, then the objective function of the original problem
\eqref{eq:19} satisfies:

\begin{equation}
\Gamma(\mathcal{S}{^{(i)}},{{\boldsymbol{\rho}}^{(i + 1)}}) \le \Gamma(\mathcal{S}{^{(i)}},{{\boldsymbol{\rho}}^{(i)}})
\label{eq:29}
\end{equation}
where the equation holds only if ${{\boldsymbol{\rho}}^{(i)}}$ is a fixed point of
\eqref{eq:19}.
\end{theorem}

\begin{proof}
Because ${g_{\rm{c}}}(\rho )$ is sign invariant (i.e.,
${g_{\rm{c}}}(\rho ) = {g_{\rm{c}}}(|\rho |)$), and concave-and-monotonically
increasing on the non-negative orthant ${\mathcal{O}_1}$, $g({\boldsymbol{\rho}})$ is sign
invariant and concave on its ${\mathcal{O}_1}$. From \eqref{eq:22}, we can get:

\begin{equation}
g({{\boldsymbol{\rho}}^{(i)}}) = f({{\boldsymbol{\rho}}^{(i)}}|{{\boldsymbol{\rho}}^{(i)}}),{\left.
{\frac{{\partial g({\boldsymbol{\rho}})}}{{\partial {\boldsymbol{\rho}}}}} \right|_{{\bf{\rho
}} = {{\boldsymbol{\rho}}^{(i)}}}} = {\left. {\frac{{\partial f({\boldsymbol{\rho}}|{{\bf{\rho
}}^{(i)}})}}{{\partial {\boldsymbol{\rho}}}}} \right|_{{\boldsymbol{\rho}} = {{\bf{\rho
}}^{(i)}}}}
\label{eq:30}
\end{equation}
And also because of the selection of ${{\bf{B}}^{(i)}}$ as \eqref{eq:26},
$f({\boldsymbol{\rho}}|{{\boldsymbol{\rho}}^{(i)}})$ is convex, majorizes $g({\boldsymbol{\rho}})$,
and shares the same tangent with $g({\boldsymbol{\rho}})$ at the point ${{\bf{\rho
}}^{(i)}}$ on ${_1}$ (as illustrated in Fig. 2). These are also true for the
other ${2^P} - 1$ orthants. We have:

\begin{equation}
f({{\boldsymbol{\rho}}^{(i + 1)}}|{{\boldsymbol{\rho}}^{(i)}}) - g({{\boldsymbol{\rho}}^{(i + 1)}})
\ge f({{\boldsymbol{\rho}}^{(i)}}|{{\boldsymbol{\rho}}^{(i)}}) - g({{\boldsymbol{\rho}}^{(i)}}) = 0
\label{eq:31}
\end{equation}
where the equation holds only if ${{\boldsymbol{\rho}}^{(i)}}$ is a fixed point, i.e.

$|{{\boldsymbol{\rho}}^{(i + 1)}}| =
|{{\boldsymbol{\rho}}^{(i)}}|$. Thus, the following relationship holds:

\begin{equation}
g({{\boldsymbol{\rho}}^{(i + 1)}}) - g({{\boldsymbol{\rho}}^{(i)}}) \le f({{\boldsymbol{\rho}}^{(i +
1)}}|{{\boldsymbol{\rho}}^{(i)}}) - f({{\boldsymbol{\rho}}^{(i)}}|{{\boldsymbol{\rho}}^{(i)}})
\label{eq:32}
\end{equation}
Because the second term of \eqref{eq:19} is the same with that of
\eqref{eq:27}, we get the conclusion of \eqref{eq:29}.
\end{proof}

\begin{figure}
\centering
\includegraphics[width=0.7\textwidth]{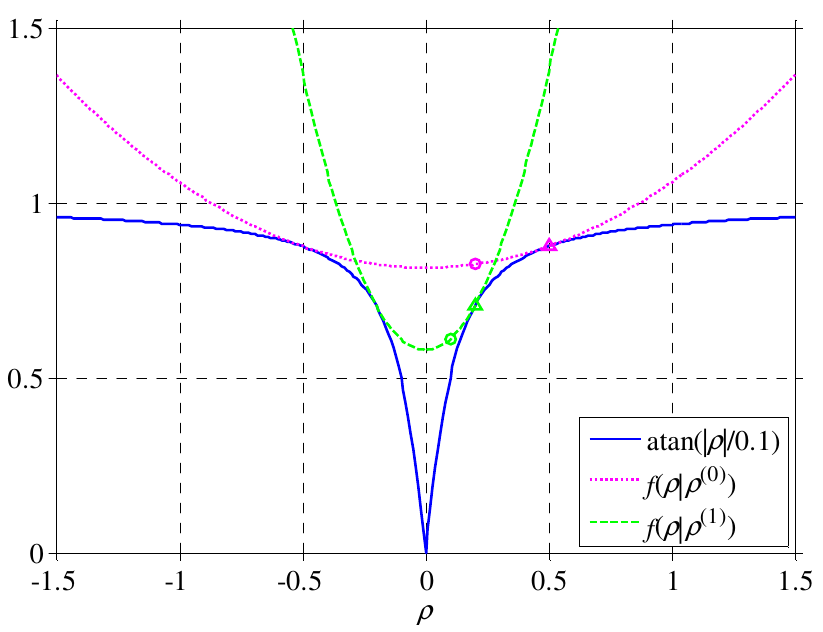}
\caption{ Illustration of the surrogate function at different iterations.}
\label{2}
\end{figure}

Theorem 2 guarantees that the surrogate function of \eqref{eq:27} can
equivalently reduce the objective function of \eqref{eq:19} in an iterative
way.

After updating ${\boldsymbol{\rho}}$ at the (\textit{i}+1)th iteration, the second step
is to adjust the parameter set, $\mathcal{S}$ , to obtain refined target positions, by
substituting \eqref{eq:28} into \eqref{eq:27}:

\begin{equation}
\begin{array}{c}
\mathop {\min }\limits_ {\mathcal{S}} \mathcal{L}(\mathcal{S}) =  - 2{\mathop{\rm Re}\nolimits} \left(
{{{{\bf{\hat r}}}^H}{\bf{\Psi }}(\mathcal{S}{^{(i)}})} \right)
\cdot {\left( {{\mathop{\rm Re}\nolimits} \left( {{{\bf{\Psi
}}^H}({^{(i)}}){\bf{\Psi }}({^{(i)}})} \right) + {\lambda}{{\bf{B}}^{(i)}}} \right)^{ - 1}}
\cdot {\mathop{\rm Re}\nolimits} \left( {{{\bf{\Psi
}}^H}(\mathcal{S}{^{(i)}}){\bf{\hat r}}} \right)
\end{array}
\label{eq:33}
\end{equation}
From \eqref{eq:33}, one can observe that it's still hard to get an solution of $\mathcal{S}$
which minimizes the objective function $\mathcal{L}(\mathcal{S})$. Fortunately, the minimization
procedure can be replaced by simply finding an appropriate  only to reduce the
objective function instead [10]. Thus, in this step we just need to find a set
${\mathcal{S}^{(i + {\rm{1}})}}$ satisfying:

\begin{equation}
\mathcal{L}(\mathcal{S}{^{(i + {\rm{1}})}}) \le \mathcal{L}(\mathcal{S}{^{(i)}})
\label{eq:34}
\end{equation}
Since $\mathcal{L}(\mathcal{S})$ is differentiable for the problem of \eqref{eq:33}, a gradient descent
method can be used to find such an estimation of ${^{(i + {\rm{1}})}}$. For every
element in ${\mathcal{S}^{(i)}}$, i.e., ${{\bf{p}}_p}$ if a gradient defined as
${\textstyle{{\partial \mathcal{L}(\mathcal{S})} \over {\partial {{\bf{p}}_p}}}}: =
{[{\textstyle{{\partial \mathcal{L}(\mathcal{S})} \over {\partial {x_p}}}},{\textstyle{{\partial \mathcal{L}(\mathcal{S})}
\over {\partial {y_p}}}},{\textstyle{{\partial \mathcal{L}(\mathcal{S})} \over {\partial {z_p}}}}]^T}$
can be calculated (see Appendix for more details), we can always find a scale
\textit{C} which makes

\begin{equation}
{\bf{p}}_p^{(i + 1)} = {\bf{p}}_p^{(i)} - C{\left. {\frac{{\partial
\mathcal{L}(\mathcal{S})}}{{\partial {{\bf{p}}_p}}}} \right|_{{{\bf{p}}_p} = {\bf{p}}_p^{(i)}}}
\label{eq:35}
\end{equation}
satisfying \eqref{eq:34}. When all the elements in ${\mathcal{S}^{(i)}}$ are updated, a new
${\mathcal{S}^{(i + {\rm{1}})}}$ is found at the (\textit{i}+1)th iteration of the off-grid
localization algorithm. The complete algorithm of the off-grid localization is
described as follows.

\begin{figure}
\centering
\includegraphics[width=0.7\textwidth]{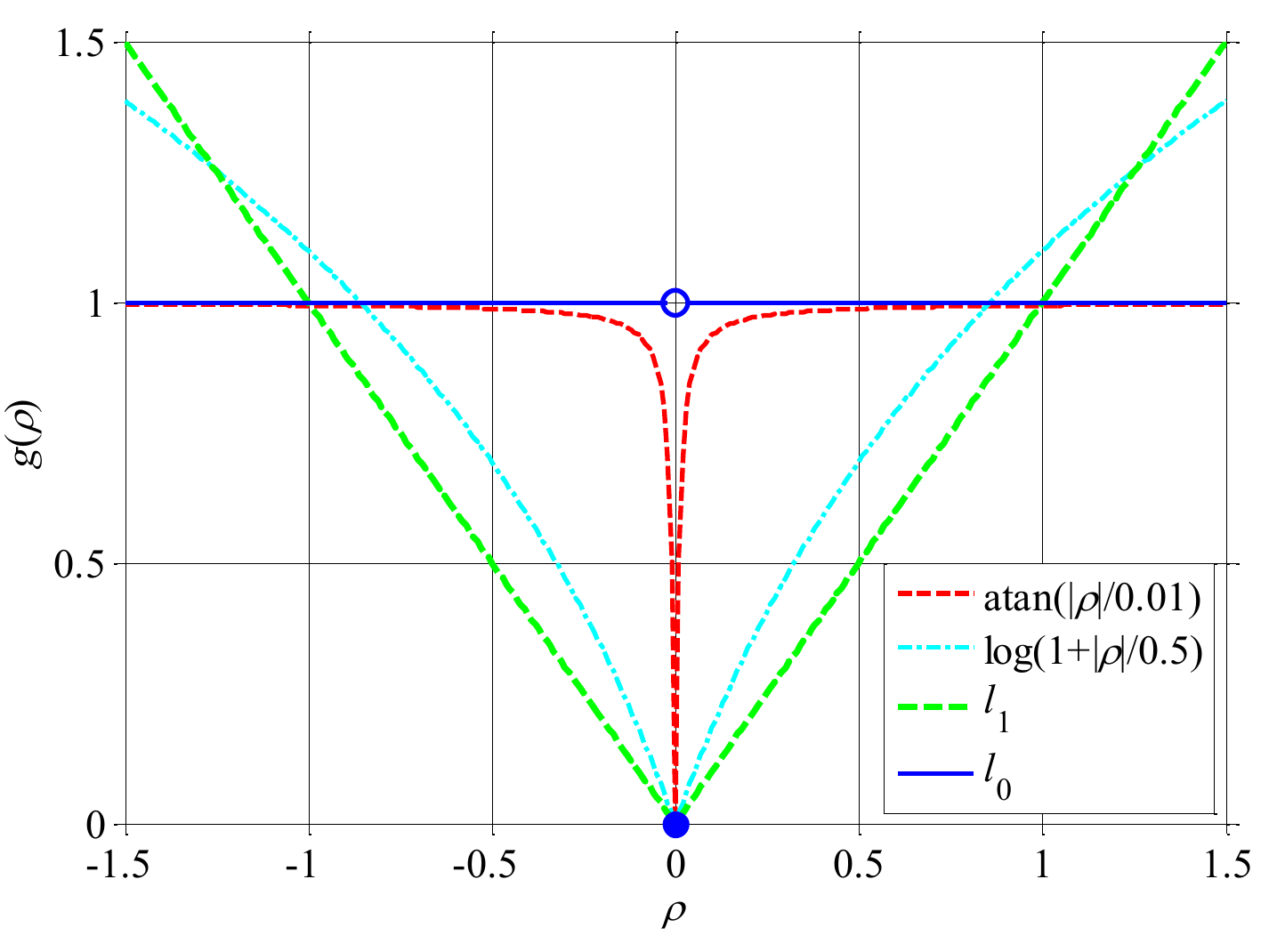}
\caption{Comparison of objective functions with a scalar variable.}
\label{3}
\end{figure}

\begin{algorithm}[H]                           
\caption{Off-grid localization algorithm based on iteratively reweighted method }
\label{alg2}
\begin{algorithmic}
\REQUIRE Set $i=0$, an initial value $\mathcal{S}^{(0)}$,$\hat{\bf{\Psi }^{(0)}}$ and $\bf{B}^{(0)}=\bf{I}$
\ENSURE $\mathcal{S} $, $\boldsymbol{\rho}$
\STATE Step 1. Estimate the power $\boldsymbol{\rho}^{(i+1)}$, of signals by (24).
\STATE Step 2. Update every element by (31) to get a new $\mathcal{S}^{(i+1)}$.
\STATE Step 3. Calculate the weight $\bf{B}^{(i+1)}$ by (22)
\STATE Step 4. Pruning the set $\mathcal{S}^{(i+1)}$
\STATE Step 5. Let $i=i+1$, and repeat Step 1 through 4 when $\|\boldsymbol{\rho}^{(i+1)}-\boldsymbol{\rho}^{(i)}\|\leq \varepsilon$
\end{algorithmic}
\end{algorithm}

\subsection{Discussions}
\subsubsection{Initialize}
The proposed MM based algorithm may still converge to a local minimum, since the
original objective function is non-convex. However, the existence of the sparse
encouraging objective function with the regularization parameter $\lambda $ and
initial set ${\mathcal{S}^{({\rm{0}})}}$ properly selected will somehow partially improve
the situation.

In the initialization step of the algorithm, the localization area can be
uniformly divided into some coarse grids only fine enough to prevent the
algorithm from falling into some unpredictable local minimizations. The coarse
grids will dramatically reduce the computational complexity.

\subsubsection{Pruning}
Since the localization is a typical sparse recovery problem, the computational complexity
can be reduced by pruning operations. At each iteration, those signals with small
powers can be pruned from the original set ${\mathcal{S}^{(i + {\rm{1}})}}$ and the vector
${{\boldsymbol{\rho}}^{(i + 1)}}$. A hard threshold $\tau $ can be used to decide which
signal needs to be removed, i.e., $\rho _p^{(i + 1)} < \tau $. Hence, the
dimension of the signal power ${{\boldsymbol{\rho}}^{(i + 1)}}$ and the position set
${^{(i + {\rm{1}})}}$ shrinks at some iterations.

By the definition of ${\ell _0}$-norm, the sparse original model has the ability to make energy gather on the certain location, but other continuous alternative methods, such as reweighted ${\ell _1}$-norm or log-sum penalty, need some strict condition of measurement matrices or lack the theoretical guarantee to ensure recover sparse solution. Compared with these methods, we have proved in Theorem 1 that the proposed method in this paper has the same ability as ${\ell _0}$-norm. Therefore, $atan|\rho /\delta |$ 
\begin{figure}
\centering
\includegraphics[width=0.7\textwidth]{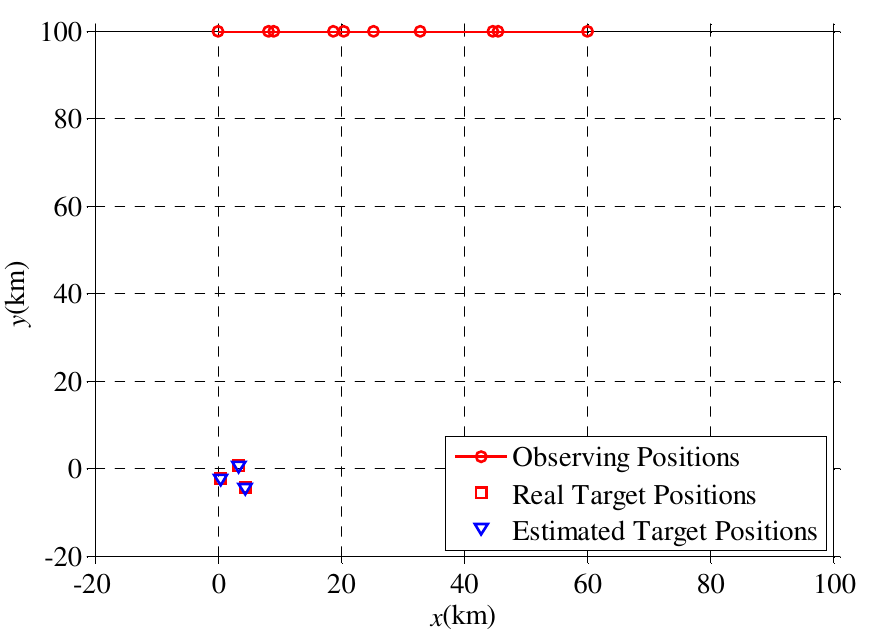}
\caption{ Localization geometry for the simulations.}
\label{4}
\end{figure}
can better approximate the ${\ell _0}$ penalty than ${\rm{log}}|\rho |$. It has
the potential to be more sparsity encouraging than the reweighted ${\ell _1}$-norm and the log-sum penalty function [19].

Fig. \ref{3} shows the shapes of different objective functions with respect to a
scalar variable $\rho $. Notice that $atan|\rho /\delta |$ is more like ${\ell
_0}$ penalty function than ${\rm{log}}|\rho |$. The smaller the parameter $\delta
 > 0$ becomes, the more similar $atan|\rho /\delta |$ is to ${\ell _0}$.

\section{Numerical Simulations}
In this section, numerical experiments are performed to proof the promise of the
proposed localization estimation model using sparse representation of array
covariance matrix.

The localization geometry is depicted in Fig. 4. The circle markers represent
the positions of the reference antenna on a moving observing platform at
different time instances. The solid line represents the trajectory of the moving
array. The length of the solid line denotes the overall virtual aperture. The
virtual aperture is the maximum distance between the observing positions. The
antenna array in the simulations is a non-uniform linear array, and the antennas
are randomly distributed along \textit{x}
\begin{figure}
\centering
\includegraphics[width=0.7\textwidth]{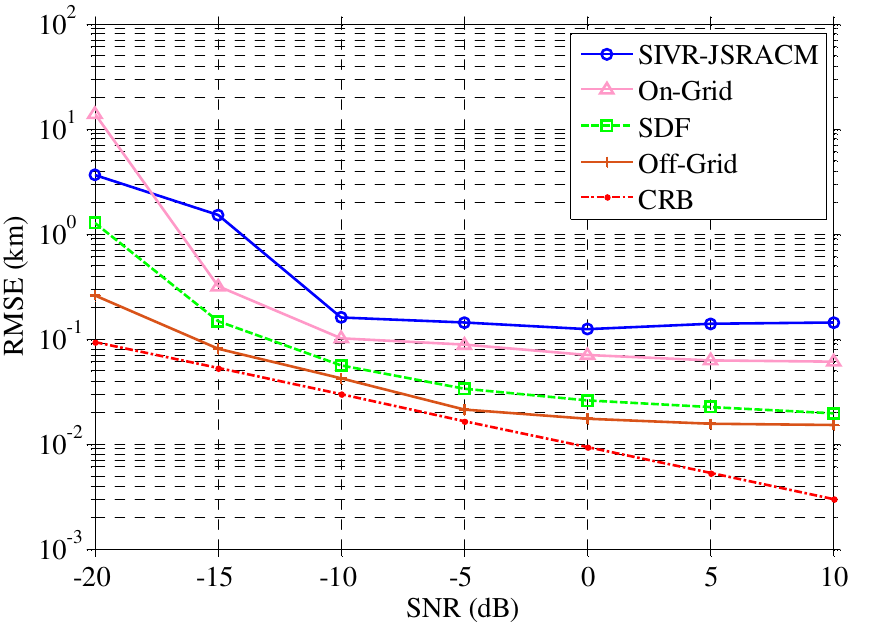}
\caption{RMSE of different methods.}
\label{5}
\end{figure}
coordinates in the range of $[0,D]$, relative to the reference antenna, where
\textit{D} is the real array aperture. The coordinates of the closely spaced
targets are [0.5140, -2.4755], [3.5106, 0.5102], [4.5115, -4.4876].

We carry out the simulation experiment to compare the performance with the
existing popular algorithms in Fig. 5. The experiment conditions are as follows.
The carrier frequencies of the signals are 6GHz. The real array aperture
\textit{D}=4m, the number of the observing points \textit{M} is 10, and the
number of the array elements \textit{L }is 11. The number of the snapshot at each
observing points is \textit{N}=1000. The performance of the localization
algorithm is measured by means of the root mean square error (RMSE), which is
defined as the average of \textit{Q} independent Monte Carlo trials, that is,

\begin{equation}
RMSE = \sqrt {\frac{1}{{P \cdot Q}}\sum\limits_{q = 1}^Q {\sum\limits_{p = 1}^P
{||{\bf{\hat p}}_p^{(q)} - {\bf{p}}_p^{(q)}|{|^2}} } }
\label{eq:37}
\end{equation}
The line with circle markers represents the JSRACM algorithm of [6], and the
line with square markers is the subspace data fusion (SDF) algorithm of [5]. The
lines with markers `triangle' and `plus' represent the proposed on-grid and
off-grid algorithm respectively. The CRB is a line with dot markers, and
\begin{figure}
\centering
\includegraphics[width=0.7\textwidth]{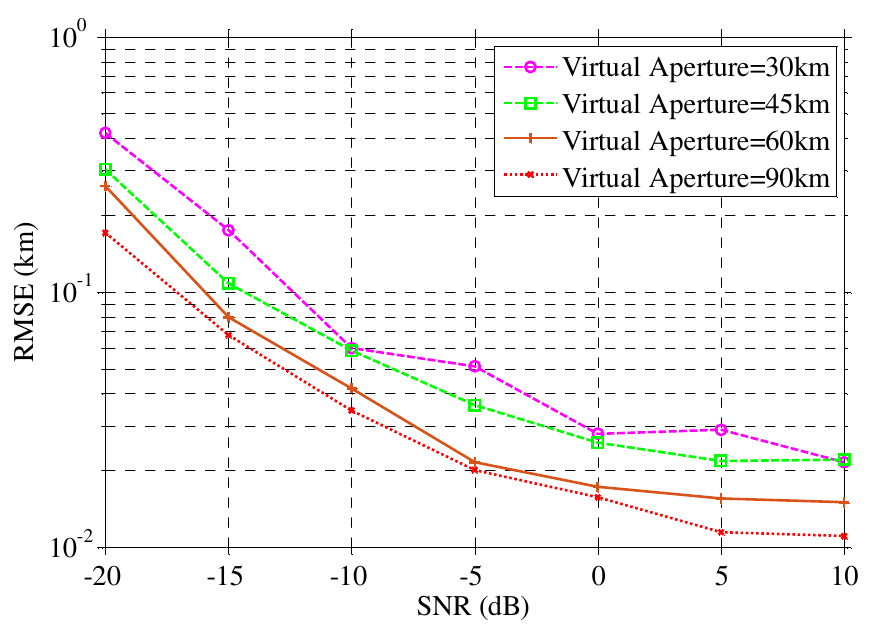}
\caption{RMSE for different virtual apertures.}
\label{7}
\end{figure}
\begin{figure}
\centering
\includegraphics[width=0.7\textwidth]{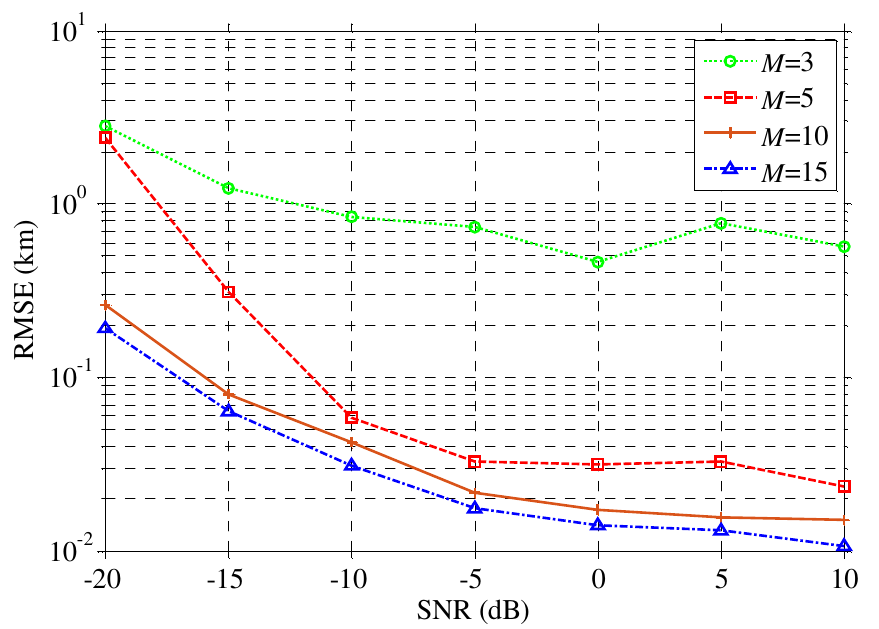}
\caption{RMSE for different numbers of observing points.}
\label{8}
\end{figure}
the calculation of CRB is given in
the Appendix B. The experimental condition is the same for each algorithm. The
grid resolution of the JSRACM and SDF algorithm is set to 25m. High grid
resolution results in very high computational complexity. As demonstrated in Fig.
5, the localization performance is greatly improved by using the proposed
off-grid algorithm.

The RMSEs of the localization for different virtual apertures are shown in Fig.\ref{7}. The experiment's condition is the same as the first one, except that the
virtual aperture varies from 30km to 90km. We observe that the localization
precision is approximately proportional to the virtual aperture. However, when
the virtual aperture is too small, the off-grid algorithm is easy to fall into a
local minimum value to fail the localization.

The number of the observing points has a considerably
\begin{figure}
\centering
\includegraphics[width=0.7\textwidth]{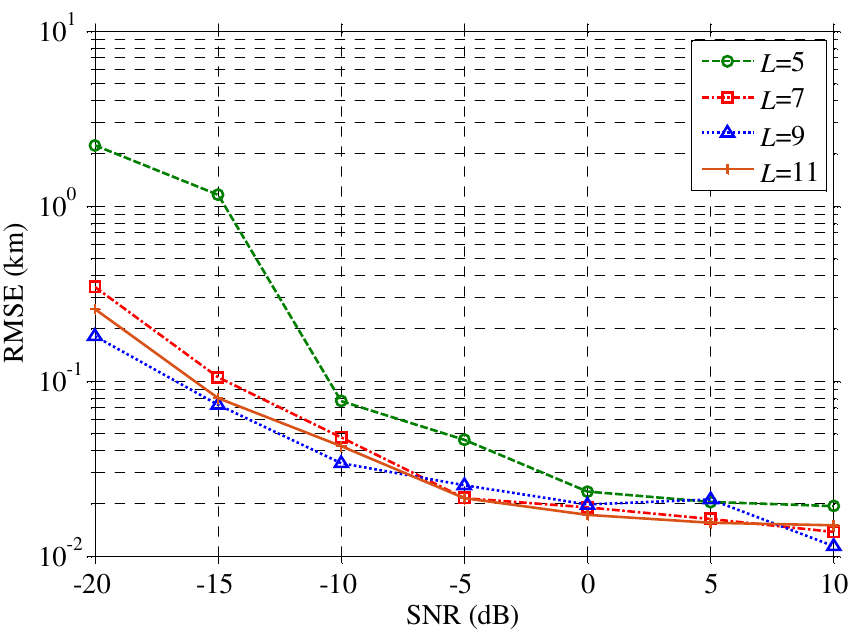}
\caption{RMSE for different numbers of antennas.}
\label{9}
\end{figure}
\begin{figure}
\centering
\includegraphics[width=0.7\textwidth]{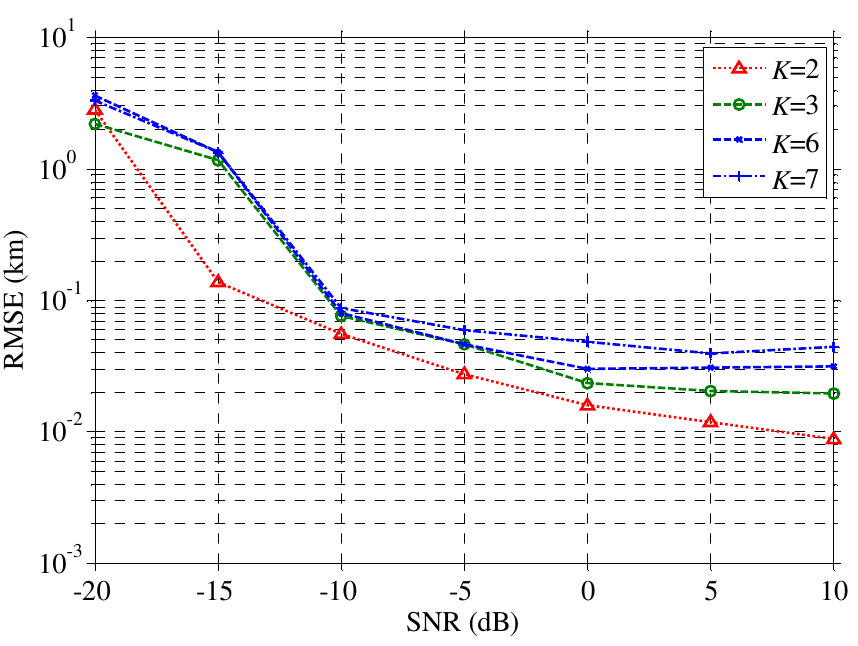}
\caption{RMSE for different numbers of targets.}
\label{10}
\end{figure}
impact on the localization precision
as in Fig.\ref{8}. A comparatively small \textit{M} is needed in the proposed
algorithm. It is only needed to be large enough to meet the demand of the
observing condition. That is a good property for most localization problems
because numerous observing sites are often unavailable. For instance, in the
passive surveillance, the intercepted signals have limited measurement samples
due to the noncooperative receiving. This property significantly outweighs the
phase difference rate algorithm [3], because the latter needs a large amount of
consecutive observing points to track the rate. As in Fig. 8, the proposed
algorithm in this paper only needs three observing points to get an acceptable
performance.

The number of the antennas also has a significant influence on the localization
precision. As in Fig.\ref{9}, the experiment's condition is the same as the first one,
except that the number of the antennas varies from five to eleven. More antennas
result in higher localization performance. However, the experiment is also
inspiring in the implementation, as the proposed algorithm does not need a great
number of array elements.

The proposed algorithm can also deal with the situation that the number of
targets is greater than that of the antennas. As shown in Fig.\ref{10}, the number of
antennas is five. The SDF algorithm [5] can only estimate the positions of no
more than four targets. We notice that with the number of targets increasing, the
precision will relatively decrease. However, the proposed off-grid algorithm
still reaches a fairly high performance even when the number of targets is
greater than that of antennas.

\section{Conclusions}
We propose a novel one-step localization system model using the sparse representation of the array covariance matrix for both on-grid and off-grid scenarios. The simulation verifies the idea that the atan-sum is more sparsity encouraging. The experiments also show that the localization performance is more dependent on the virtual aperture, and less dependent on the number of the array elements, and the number of the observing sites. The ability of the multiple-target localization is also verified.

\appendix

\section{Proof of Theorem 1}
In order to show our proof clearly, we use $\bf{\Psi }={\bf{\Psi }}(\mathcal{S})$ for a given $\mathcal{S}$. By reference to the mathematical skill in \cite{x5}, we will give a rigorous proof of the equivalence between the alternative method and the original sparse problem.  

In sparse recovery problem, the core problem is to solve the following $\ell_0$-minimization problem,
\begin{eqnarray}\label{l0}
\min \limits_{{\bf{x}}\in \mathbb{R}^n}\|{\bf{x}}\|_0 \ \ s.t. \ {\bf{\Psi}}{\bf{x}}={\bf{b}}
\end{eqnarray}
Because $\ell_0$-minimization is a NP-HARD problem, so we consider the following alternative function $g_\delta(\cdot)$ to replace $\ell_0$-norm, and we can derive the corresponding $\ell_{g_\delta}$-minimization,
\begin{eqnarray}\label{fp-model}
\min\limits_{{\bf{x}}\in \mathbb{R}^n} \|{\bf{x}}\|_{g_\delta} s.t.\ {\bf{\Psi}}{\bf{x}}={\bf{b}}
\end{eqnarray}
where $\|{\bf{x}}\|_{g_\delta}=\sum \limits_{i=1}^n g_\delta(x_i)$.

In practical application, the measurement vector ${\bf{b}}$ usually contain noise, so model \ref{l0} and model \ref{fp-model} can be changed into the following regularization models
\begin{eqnarray}\label{l0-r}
\min \|{\bf{\Psi}}{\bf{x}}-{\bf{b}}\|_2^2+\lambda\|{\bf{x}}\|_0
\end{eqnarray}
and
\begin{eqnarray}\label{fp-r}
\min \|{\bf{\Psi}}{\bf{x}}-{\bf{b}}\|_2^2+\lambda\|{\bf{x}}\|_{g_\delta}
\end{eqnarray}
Therefore, the main contribution of this paper is to prove the equivalence relationship between model(\ref{l0-r}) and model(\ref{fp-r}).
\begin{theorem}\label{theorem3}\cite{x2}
There exists a constant $\lambda_0>0$, such that the minimization problem \ref{l0} and \ref{l0-r} have the same solution for all $0<\lambda<\lambda_0$.
\end{theorem}
Next, we will prove the equivalence between $l_{g_\delta}$-minimization and $l_0$-minimization.
\begin{lemma}\label{lemma2-1}
If ${\bf{x}}^*$ is solution of $l_{g_\delta}$-minimization, then the sub-matrix ${\bf{\Psi}}_S$ is column full rank, where $S=supp({\bf{x}}^*)$
\end{lemma}
\begin{proof}
If there exists a vector ${\bf{h}} \in Ker({\bf{\Psi}})$ such that $supp({\bf{h}})\subseteq S$. Without of generality, we note that ${\bf{x}}^*=[x^*_1,x^*_2,...,x^*_n]^T$. Let
\begin{equation}
a=\max \limits_{h_i\neq 0} \frac{|x^*_i|}{|h_i|}
\end{equation}
therefore, it is easy to proof that
\begin{equation}
\notag sgn(x^*_i+\alpha h_i)=sgn(x^*_i-\alpha h_i)=sgn(x^*_i)
\end{equation}
for $\alpha \in [0,a]$. Since $g_\delta({\bf{x}})$ is a concave function when ${\bf{x}}\geq 0$ or ${\bf{x}}\leq 0$, it is easy to get that
\begin{eqnarray}
\notag g_\delta(|x^*_i|)&=&g_\delta(\frac{1}{2}|x^*_i+\alpha h_i|+\frac{1}{2}|x^*_i-\alpha h_i|)\\
\notag &> & \frac{1}{2}g_\delta(|x^*_i+\alpha h_i|)+\frac{1}{2}g_\delta(|x^*_i-\alpha h_i|)
\end{eqnarray}
Therefore,
\begin{equation}
\notag \|{\bf{x}}^*\|_{g_\delta}> \frac{1}{2}\|{\bf{x}}^*+\alpha {\bf{h}}\|_{g_\delta}+\frac{1}{2}\|{\bf{x}}^*-\alpha {\bf{h}}\|_{g_\delta}
\end{equation}
which contradicts the assumption condition.
\end{proof}
\begin{theorem}\label{theorem4}
There exists a constant $\delta^*({\bf{\Psi}},{\bf{b}})$ based on ${\bf{\Psi}}$ and ${\bf{b}}$ such that the solution of $l_{g_\delta}$-minimization also solves $l_0$-minimization whenever $0<\delta<\delta^*({\bf{\Psi}},{\bf{b}})$.
\end{theorem}
\begin{proof}
For fixed ${\bf{\Psi}}$ and ${\bf{b}}$, we can define the following sets,
\begin{eqnarray}
V=\left\{{\bf{x}}| {\bf{\Psi}}{\bf{x}}={\bf{b}},{\bf{\Psi}}_S\text{ is column full rank, where }S=supp({\bf{x}})\right\}
\end{eqnarray}
and
\begin{eqnarray}
W=\left\{{\bf{x}}\in V|\forall {\bf{y}}\in V, \|{\bf{x}}\|_0\leq \|{\bf{y}}\|_0 \right\}
\end{eqnarray}
It is obvious that $W\subseteq V$ and easy to get the following inequality,
\begin{eqnarray}
\|{\bf{x}}\|_0\leq \|{\bf{y}}\|_0+1
\end{eqnarray}
for any ${\bf{x}}\in W$ and ${\bf{y}}\in \overline{W}$, and there exists a constant $\delta({\bf{x}},{\bf{y}})$ such that
\begin{eqnarray}
\|{\bf{x}}\|_{g_\delta}<\|{\bf{y}}\|_{g_\delta}
\end{eqnarray}
whenever $0<\delta<\delta({\bf{x}},{\bf{y}})$. Since the set $V$ is finite, for any ${\bf{x}}\in W$ and ${\bf{y}}\in \overline{W}$, we can conclude that
\begin{eqnarray}\label{0417_1008}
\|{\bf{x}}\|_{g_\delta}<\|{\bf{y}}\|_{g_\delta}
\end{eqnarray}
whenever $0<\delta<\delta^*({\bf{\Psi}},{\bf{b}})$, where
\begin{eqnarray}
\delta^*({\bf{\Psi}},{\bf{b}})=\min \limits_{{\bf{x}}\in W,{\bf{y}}\in \overline{W}} \delta({\bf{x}},{\bf{y}}).
\end{eqnarray}
By the definition of $0$-norm and Lemma (\ref{lemma2-1}), it is obvious that the solutions of $l_0$-minimization belongs to the set $W$ and we can get conclusion of this theorem by (\ref{0417_1008}).
\end{proof}
\begin{lemma}\label{lemma3-1}
If ${\bf{x}}^*$ is the solution of $l_{g_\delta}^{\lambda}$-minimization, for any ${\bf{h}}$ with $supp({\bf{h}})\subseteq supp({\bf{x}}^*)$, we have that
\begin{eqnarray}
2<{\bf{b}}-{\bf{\Psi}}{\bf{x}}^*,{\bf{\Psi}}{\bf{h}}>=\lambda\sum \limits_{i\in supp({\bf{x}}^*)}\frac{\delta h_i sgn(x^*_i)}{\delta^2+(x_i^*)^2}
\end{eqnarray}
and
\begin{eqnarray}
2({\bf{\Psi}}^T({\bf{b}}-{\bf{\Psi}}{\bf{x}}^*))_i=\frac{\lambda \delta sgn(x_i^*)}{\delta^2+(x_i^*)^2}
\end{eqnarray}
for any $i\in supp({\bf{x}}^*)$.
\end{lemma}
\begin{proof}
For any $t\in \mathbb{R}$ and ${\bf{h}}\in \mathbb{R} ^n$ with $supp({\bf{h}})\subseteq supp({\bf{x}}^*)$, it is obvious that
\begin{eqnarray}
\|{\bf{\Psi}}{\bf{x}}^*-{\bf{b}}\|_2^2+\lambda\|{\bf{x}}^*\|_{g_\delta}\leq \|{\bf{\Psi}}({\bf{x}}^*+t{\bf{h}})-{\bf{b}}\|_2^2+\lambda\|{\bf{x}}^*+t{\bf{h}}\|_{g_\delta}
\end{eqnarray}
Since $supp({\bf{h}})\subseteq supp({\bf{x}}^*)$, we can get that
\begin{eqnarray}
\|{\bf{x}}^*+t{\bf{h}}\|_{g_\delta}-\|{\bf{x}}^*\|_{g_\delta}=\sum \limits_{i\in supp({\bf{x}}^*)}\left(atan\frac{|x^*_i+t_i|}{\delta}-atan\frac{|x^*_i|}{\delta}\right)
\end{eqnarray}
Therefore, we can get that
\begin{eqnarray}
t^2\|{\bf{\Psi}}{\bf{h}}\|_2^2+2t<{\bf{\Psi}}{\bf{x}}^*-{\bf{b}},{\bf{\Psi}}{\bf{h}}>+\lambda\sum \limits_{i\in supp({\bf{x}}^*)}\left(atan\frac{|x^*_i+th_i|}{\delta}-atan\frac{|x^*_i|}{\delta}\right)\geq 0
\end{eqnarray}
Let $t>0$ and $t\rightarrow 0^+$, we can get that
\begin{eqnarray}
2<{\bf{\Psi}}{\bf{x}}^*-{\bf{b}},{\bf{\Psi}}{\bf{h}}>+\lambda \lim \limits_{t\rightarrow 0^+} \sum \limits_{i\in supp({\bf{x}}^*)}\frac{atan\frac{|x^*_i+th_i|}{\delta}-atan\frac{|x^*_i|}{\delta}}{t}\geq 0,
\end{eqnarray}
i.e.,
\begin{eqnarray}
2<{\bf{\Psi}}{\bf{x}}^*-{\bf{b}},{\bf{\Psi}}{\bf{h}}>+\lambda \sum \limits_{i\in supp({\bf{x}}^*)}\frac{\delta \cdot sgn(x_i^*)}{\delta^2+(x_i^*)^2}\geq 0
\end{eqnarray}
Let $t<0$ and $t\rightarrow 0^-$, repeat the above action, we can get that
\begin{eqnarray}
2<{\bf{\Psi}}{\bf{x}}^*-{\bf{b}},{\bf{\Psi}}{\bf{h}}>+\lambda \sum \limits_{i\in supp({\bf{x}}^*)}\frac{\delta \cdot sgn(x_i^*)}{\delta^2+(x_i^*)^2}\leq 0
\end{eqnarray}
Therefore, we can get that
\begin{eqnarray}
2<{\bf{b}}-{\bf{\Psi}}{\bf{x}}^*,{\bf{\Psi}}{\bf{h}}>=\lambda\sum \limits_{i\in supp({\bf{x}}^*)}\frac{\delta \cdot h_isgn(x^*_i)}{\delta^2+(x_i^*)^2}
\end{eqnarray}
and we can get the second conclusion of this lemma if ${\bf{h}}={\bf{e}}_i$.
\end{proof}
\begin{lemma}\label{lemma3-2}
If ${\bf{x}}^*$ is the solution of $l_{g_\delta}^{\lambda}$-minimization and $\lambda > \frac{2\|{\bf{b}}\|_2^2}{\pi}$, then we have that
\begin{eqnarray}
\|{\bf{x}}^*\|_{\infty}\leq \delta\cdot tan\left( \frac{\|{\bf{b}}\|_2^2}{\lambda}\right).
\end{eqnarray}
\end{lemma}
\begin{proof}
Since ${\bf{x}}^*$ is the solution of $l_{g_\delta}^{\lambda}$-minimization, we can get that
\begin{eqnarray}
\|{\bf{\Psi}}{\bf{x}}^*-{\bf{b}}\|_2^2+\lambda\|{\bf{x}}^*\|_{g_\delta}\leq \|{\bf{b}}\|_2^2
\end{eqnarray}
therefore, we have that
\begin{eqnarray}
\lambda\|{\bf{x}}^*\|_{g_\delta}\leq \|{\bf{b}}\|_2^2
\end{eqnarray}
Since $\lambda > \frac{2\|{\bf{b}}\|_2^2}{\pi}$, so it is easy to get that
\begin{eqnarray}
atan \frac{\|{\bf{x}}^*\|_{\infty}}{\delta}\leq \|{\bf{b}}\|_2^2
\end{eqnarray}
and
\begin{eqnarray}
\|{\bf{x}}^*\|_{\infty}\leq \delta\cdot tan \left(\frac{\|{\bf{b}}\|_2^2}{\lambda}\right).
\end{eqnarray}
\end{proof}
\begin{theorem}\label{theorem5}
There exists a constant $\lambda > \frac{2\|{\bf{b}}\|_2^2}{\pi}$ and if the following inequality holds
\begin{eqnarray}
\delta^2\frac{\lambda ^3}{4}m\|{\bf{\Psi}}\|_2^4\left(1+\frac{\|{\bf{b}}\|_2^2}{\lambda}\right)<\sigma_{min}({\bf{\Psi}})
\end{eqnarray}
where
\begin{eqnarray}
\notag \sigma_{min}({\bf{\Psi}})=min\{\sigma({\bf{\Psi}})|&&\sigma({\bf{\Psi}}) \ is \ the \ smallest \ singular \ value \ of \ the \ sub-matrix\\ \notag && which \ is \ composed \ by \ the \ linearly \ independent \\
&& column \ vectors \ of \ {\bf{\Psi}}\}
\end{eqnarray}
then the solution of $l_{g_\delta}^{\lambda}$-minimization also solves $l_{g_\delta}$-minimization.
\end{theorem}
\begin{proof}
We assume that the solution ${\bf{x}}^*$ of $l_{g_\delta}^{\lambda}$-minimization is not the solution of $l_{g_\delta}$-minimization, it is easy to get that
\begin{eqnarray}
{\bf{\Psi}}{\bf{x}}^*={\bf{b}}^*\neq {\bf{b}}
\end{eqnarray}
Let ${\bf{y}}^*$ be the sparest solution of ${\bf{A}}{\bf{y}}={\bf{b}}-{\bf{b}}^*$. Let ${\bf{B}}={\bf{\Psi}}_{supp({\bf{y}}^*)}$ and $\hat{{\bf{y}}}^*={\bf{y}}^*_{supp({\bf{y}}^*)}$. It is obvious that
\begin{eqnarray}
\sigma_{min}({\bf{\Psi}})\leq \frac{\|{\bf{B}}\hat{{\bf{y}}}^*\|_2^2}{\|\hat{{\bf{y}}}^*\|_2^2}
\end{eqnarray}
On the other hand, we have that
\begin{eqnarray}
\notag \|{\bf{y}}^*\|_2^2 &\geq & \frac{\|{\bf{A}}{\bf{y}}\|_2^2}{\|{\bf{\Psi}}\|_2^2}\\
&\geq &  \frac{\|{\bf{\Psi}}^T({\bf{\Psi}}{\bf{x}}^*-{\bf{b}})\|_2^2}{\|{\bf{A}}{\bf{y}}\|_2^4}
\end{eqnarray}
By Lemma \ref{lemma3-1} and Lemma \ref{lemma3-2}, we can get that
\begin{eqnarray}
\notag \|{\bf{\Psi}}^T({\bf{\Psi}}{\bf{x}}^*-{\bf{b}})\|_2^2 & \geq & \sum \limits_{i\in supp({\bf{x}}^*)}\frac{\lambda^2}{4}\cdot \frac{\delta^2}{(\delta^2+(x^*_i)^2)^2} \\
\notag &\geq & \frac{\lambda^2}{4}\cdot \frac{\delta^2}{(\delta^2+\|{\bf{x}}^*\|_{\infty}^2)^2} \\
\notag & \geq & \frac{\lambda^2}{4}\cdot \frac{1}{\delta^2\left(1+tan^2\frac{\|{\bf{b}}\|_2^2}{\lambda}\right)}
\end{eqnarray}
Therefore, we can get that
\begin{eqnarray}
\|{\bf{y}}^*\|_2^2\geq \frac{1}{\|{\bf{\Psi}}\|_2^4}\cdot \frac{\lambda ^2}{4}\frac{1}{\delta^2\left(1+tan^2\frac{\|{\bf{b}}\|_2^2}{\lambda}\right)}
\end{eqnarray}
By the assumption condition, we can get that
\begin{eqnarray}
\notag \lambda \|{\bf{y}}^*\|_{g_\delta} &\leq & \lambda m \cdot \frac{\|{\bf{y}}^*\|_2^2}{\|{\bf{y}}^*\|_2^2} \\
\notag &\leq & \lambda m\|{\bf{\Psi}}\|_2^4\cdot\frac{\lambda ^2}{4}\delta^2\left(1+tan^2\frac{\|{\bf{b}}\|_2^2}{\lambda}\right)\cdot \|{\bf{y}}^*\|_2^2 \\
\notag &\leq & \sigma_{min}({\bf{\Psi}})\|{\bf{y}}^*\|_2^2 \\
\notag &\leq &\frac{\|{\bf{b}}*\hat{{\bf{y}}}^*\|_2^2}{\|{\bf{y}}^*\|_2^2}\|{\bf{y}}^*\|_2^2 \\
&\leq & \|{\bf{b}}^*\hat{{\bf{y}}}^*\|_2^2=\|{\bf{\Psi}}{\bf{x}}^*-{\bf{b}}\|_2^2
\end{eqnarray}
So we can get that
\begin{eqnarray}
\notag \|{\bf{\Psi}}({\bf{x}}^*+{\bf{y}}^*)\|_2^2+\lambda\|{\bf{x}}^*+{\bf{y}}^*\|_{g_\delta}&=& \lambda \|{\bf{x}}^*+{\bf{y}}^*\|_{g_\delta} \\
\notag &\leq & \lambda \|{\bf{x}}^*\|_{g_\delta}+\lambda \|{\bf{y}}^*\|_{g_\delta} \\
& < &\lambda \|{\bf{x}}^*\|_{g_\delta}+\|{\bf{\Psi}}{\bf{x}}-{\bf{b}}\|_2^2
\end{eqnarray}
which contradicts the assumption.
\end{proof}
By Theorem \ref{theorem3}, Theorem \ref{theorem4} and Theorem \ref{theorem5}, we can get the equivalence between $l_{g_\delta}^{\lambda}$-minimization (\ref{fp-r}), $l_{0}^{\lambda}$-minimization (\ref{l0-r}), which means Theorem 1 is proved.

\section{Calculation of the Gradient}
To calculate the derivation of the cost function $\mathcal{L}(\mathcal{S})$ with respect to the
coordinate, ${x_p}$, of the target position, we rewrite $\mathcal{L}(\mathcal{S})$ as:

\begin{equation}
\mathcal{L}(\mathcal{S}) =  - {{\bf{u}}^T}{{\bf{V}}^{ - 1}}{\bf{u}}
\label{eq:38}
\end{equation}
where
\begin{equation}
{\bf{u}}: = 2{\mathop{\rm Re}\nolimits} \left( {{{\bf{\Psi
}}^H}(\mathcal{S}){\bf{\hat r}}} \right) \in {^P}
\label{eq:39}
\end{equation}
and

\begin{equation}
{\bf{V}}: = 2\left( {{\mathop{\rm Re}\nolimits} \left( {{{\bf{\Psi
}}^H}(\mathcal{S}){\bf{\Psi }}(\mathcal{S})} \right) + {\lambda }{\bf{B}}} \right)
\in {^{P \times P}}
\label{eq:40}
\end{equation}
Referring to [22], and using the chain rule, ${\textstyle{{\partial ()} \over
{\partial {x_p}}}}$ can be expressed as:

\begin{equation}
\frac{{\partial }\mathcal{L}}{{\partial {x_p}}} = (\Delta{_{\bf{\Psi }}}\mathcal{L})(\Delta{_{{x_p}}}{\bf{\Psi }})
+ (\Delta{_{{{\bf{\Psi }}^*}}}\mathcal{L})(\Delta{_{{x_p}}}{{\bf{\Psi }}^*})
\label{eq:41}
\end{equation}
where $ \Delta $ denotes the vectorized version of the derivation of a complex matrix with
respect to another complex matrix, and $\Delta{_{\bf{\Psi }}}\mathcal{L} \in {^{{L^2}MP}}$,
$\Delta{_{{x_p}}}{\bf{\Psi }} \in {^{{L^2}MP}}$. Notice that the notation, $\mathcal{S}{^{(i)}}$,
in \eqref{eq:42} is omitted for the sake of simplicity. Because $ \mathcal{L} $ and ${x_p}$
are real variables, we get:

\begin{equation}
\frac{{\partial }\mathcal{L}}{{\partial {x_p}}} = 2{\mathop{\rm Re}\nolimits} \left[
{(\Delta{_{\bf{\Psi }}}\mathcal{L})(\Delta{_{{x_p}}}{\bf{\Psi }})} \right]
\label{eq:42}
\end{equation}
Using the lemma of finding the derivative of a product of two functions,
$\Delta{_{\bf{\Psi }}}\mathcal{L}$ can be calculated as:

\begin{equation}
\begin{array}{l}
{\Delta_{\bf{\Psi }}\mathcal{L}} =  - ({({{\bf{V}}^{ - 1}}{\bf{u}})^T} \otimes
{{\bf{I}}_1})\Delta{_{\bf{\Psi }}}{{\bf{u}}^T} - ({{\bf{I}}_1} \otimes
{{\bf{u}}^T})\Delta{_{\bf{\Psi }}}({{\bf{V}}^{ - 1}}{\bf{u}})\\
=  - {({{\bf{V}}^{ - 1}}{\bf{u}})^T}\Delta{_{\bf{\Psi }}}{{\bf{u}}^T} -
{{\bf{u}}^T}(({{\bf{u}}^T} \otimes {{\bf{I}}_P})\Delta{_{\bf{\Psi }}}{{\bf{V}}^{ - 1}}
+ ({{\bf{I}}_1} \otimes {{\bf{V}}^{ - 1}})\Delta{_{\bf{\Psi }}}{\bf{u}})\\
=  - {({{\bf{V}}^{ - 1}}{\bf{u}})^T}\Delta{_{\bf{\Psi }}}{{\bf{u}}^T} -
{{\bf{u}}^T}({{\bf{u}}^T} \otimes {{\bf{I}}_P})\Delta{_{\bf{\Psi }}}{{\bf{V}}^{ - 1}} -
{{\bf{u}}^T}{{\bf{V}}^{ - 1}}\Delta{_{\bf{\Psi }}}{\bf{u}}
\end{array}
\label{eq:43}
\end{equation}

{\raggedright
where
}

\begin{equation}
{\Delta_{\bf{\Psi }}}{\bf{u}} = {\Delta_{\bf{\Psi }}}{{\bf{u}}^T} = {{\bf{I}}_P} \otimes
({{\bf{\hat r}}^H})
\label{eq:44}
\end{equation}

{\raggedright
and
}

\begin{equation}
\begin{array}{c}
{\Delta_{\bf{\Psi }}}{{\bf{V}}^{ - 1}} =  - {({{\bf{V}}^T})^{ - 1}} \otimes
{{\bf{V}}^{ - 1}}{\Delta_{\bf{\Psi }}}{\bf{V}}\\
=  - \left( {{{({{\bf{V}}^T})}^{ - 1}} \otimes {{\bf{V}}^{ - 1}}} \right)\left[
{{{\bf{I}}_P} \otimes ({{\bf{\Psi }}^H})} \right.\\
+ \left. {\left( {({{\bf{\Psi }}^H}) \otimes {{\bf{I}}_P}}
\right){{\bf{K}}_{{L^2},P}}} \right]
\end{array}
\label{eq:45}
\end{equation}

{\raggedright
where ${{\bf{K}}_{{L^2},P}}$ is commutation matrix of size ${L^2}P \times
{L^2}P$ [22]. Then, ${\Delta_{\bf{\Psi }}\mathcal{L}}$ is given by the following expression:
}

\begin{equation}
\begin{array}{c}
{\Delta_{\bf{\Psi }}\mathcal{L}} =  - 2{({{\bf{V}}^{ - 1}}{\bf{u}})^T}({{\bf{I}}_P} \otimes
({{{\bf{\hat r}}}^H}))\\
+ {{\bf{u}}^T}({{\bf{u}}^T} \otimes {{\bf{I}}_P})({({{\bf{V}}^T})^{ - 1}}
\otimes {{\bf{V}}^{ - 1}})\left[ {{{\bf{I}}_P} \otimes ({{\bf{\Psi
}}^H})} \right.\\
+ \left. {\left( {({{\bf{\Psi }}^H}) \otimes {{\bf{I}}_P}}
\right){{\bf{K}}_{{L^2},P}}} \right]
\end{array}
\label{eq:46}
\end{equation}

{\raggedright
Substituting \eqref{eq:46} into \eqref{eq:42}, we get:
}

\begin{equation}
\begin{array}{c}
\frac{{\partial }\mathcal{L}}{{\partial {x_p}}} = 2{\mathop{\rm Re}\nolimits} \left\{ { -
2{{({{\bf{V}}^{ - 1}}{\bf{u}})}^T}\left( {{{\bf{I}}_P} \otimes ({{{\bf{\tilde
r}}}^H})} \right){\Delta_{{x_p}}}{\bf{\Psi }}} \right.\\
+ {{\bf{u}}^T}({{\bf{u}}^T} \otimes {{\bf{I}}_P})\left( {{{({{\bf{V}}^T})}^{ -
1}} \otimes {{\bf{V}}^{ - 1}}} \right)\left[ {{{\bf{I}}_P} \otimes ({{\bf{\Psi
}}^H})} \right.\\
\left. { + \left. {\left( {({{\bf{\Psi }}^H}) \otimes
{{\bf{I}}_P}} \right){{\bf{K}}_{{L^2},P}}} \right]{_{{x_p}}}{\bf{\Psi }}}
\right\}
\end{array}
\label{eq:47}
\end{equation}

{\raggedright
Because ${\Delta_{{x_p}}}{\bf{\Psi }} = vec({\textstyle{{\partial {\bf{\Psi }}} \over
{\partial {x_p}}}})$, and ${\bf{V}}$ is symmetric, ${\textstyle{{\partial }\mathcal{L} \over
{\partial {x_p}}}}$ can be expressed in a more compact way:
}

\begin{equation}
\begin{array}{l}
\frac{{\partial }\mathcal{L}eq:52}{{\partial {x_p}}} = 2{\mathop{\rm Re}\nolimits} \left[ { -
2{{\bf{u}}^T}{{\bf{V}}^{ - 1}}\frac{{\partial {{\bf{\Psi }}^T}}}{{\partial
{x_p}}}{{{\bf{\hat r}}}^{\rm{*}}}} \right.\\
\left. { + {{\bf{u}}^T}{{\bf{V}}^{ - 1}}({{\bf{\Psi
}}^H}\frac{{\partial {\bf{\Psi }}}}{{\partial {x_p}}} +
\frac{{\partial {{\bf{\Psi }}^T}}}{{\partial
{x_p}}}{{\bf{\Psi }}^{\rm{*}}}){{\bf{V}}^{ -
1}}{\bf{u}}} \right]
\end{array}
\label{eq:48}
\end{equation}

{\raggedright
where ${\textstyle{{\partial {\bf{\Psi }}} \over {\partial {x_p}}}} \in
{^{{L^2}M \times P}}$. Since ${x_p}$ is only relevant to the \textit{p}th column
of ${\bf{\Psi }}$, we get ${\textstyle{{\partial {\bf{\Psi }}} \over {\partial
{x_p}}}} = [{\bf{0}},...,{\textstyle{{\partial {{\boldsymbol{\psi }}_p}} \over {\partial
{x_p}}}},...,{\bf{0}}]$, where ${{\boldsymbol{\psi }}_p}$ is the \textit{p}th column of
${\bf{\Psi }}$, and ${\textstyle{{\partial {{\boldsymbol{\psi }}_p}} \over {\partial
{x_p}}}}$ is given by:
}

\begin{equation}
\frac{{\partial {{\boldsymbol{\psi }}_p}}}{{\partial {x_p}}} = \left[
{\begin{array}{*{20}{c}}
{{{\bf{\beta }}_1}(p)}\\
\cdots \\
{{{\bf{\beta }}_M}(p)}
\end{array}} \right]
\label{eq:49}
\end{equation}
The entries of \eqref{eq:49} is given by:

\begin{equation}
{{\bf{\beta }}_m}(p) = {\left[ {\beta _{11}^{(m)}(p),...,\beta
_{1L}^{(m)}(p),{\kern 1pt} {\kern 1pt} {\kern 1pt} ...,\beta
_{L1}^{(m)}(p),...,\beta _{LL}^{(m)}(p)} \right]^T}
\label{eq:50}
\end{equation}

{\raggedright
where, according to the definition of ${\bf{\Psi }}$ in \eqref{eq:10},
}

\begin{equation}
\begin{array}{c}
\beta _{lk}^{(m)}(p) =  - j2\pi \frac{f}{c}{e^{ - j2\pi {\textstyle{f \over
c}}(||{{\bf{a}}_{ml}} - {{\bf{p}}_p}|| - ||{{\bf{a}}_{mk}} -
{{\bf{p}}_p}||)}}\left[ {||{{\bf{a}}_{ml}} - {{\bf{p}}_p}|{|^{ - 1}}} \right.\\
\left. { \cdot ({x_p} - {x_{ml}}) - ||{{\bf{a}}_{mk}} - {{\bf{p}}_p}|{|^{ -
1}}({x_p} - {x_{mk}})} \right]
\end{array}
\label{eq:51}
\end{equation}
Substituting ${\textstyle{{\partial {\bf{\Psi }}} \over {\partial {x_p}}}}$ back
into \eqref{eq:48}, we get the final ${\textstyle{{\partial } \over {\partial
{x_p}}}}$. The other derivations, ${\textstyle{{\partial ()} \over {\partial
{y_p}}}}$ and ${\textstyle{{\partial ()} \over {\partial {z_p}}}}$, can also be
calculated through the same way of ${\textstyle{{\partial ()} \over {\partial
{x_p}}}}$.

\section{Cram\'{e}r-Rao Lower Bound}

{\raggedright
\hspace{2em}First we stack the signals of \textit{M} time instances all together,
}

\begin{equation}
{{\bf{s}}_{{\rm{r}}n}} = {\bf{A}}{{\bf{s}}_n} + {{\bf{n}}_{{\rm{r}}n}}
\label{eq:52}
\end{equation}

{\raggedright
where
}

\begin{equation}
{\bf{A}} = \left[ {\begin{array}{*{20}{c}}
{{\bf{{\rm A}}}({t_{\rm{1}}})}&{}&{}\\
{}& \ddots &{}\\
{}&{}&{{\bf{{\rm A}}}({t_M})}
\end{array}} \right]
\label{eq:53}
\end{equation}
the received signal vector ${{\bf{s}}_{{\rm{r}}n}} = {[\begin{array}{*{20}{c}}
{{\bf{s}}_{\rm{r}}^T({t_{1,n}})}& \cdots &{{\bf{s}}_{\rm{r}}^T({t_{M,n}})}
\end{array}]^T}$, the emitter signal vector ${{\bf{s}}_n} =
{[\begin{array}{*{20}{c}}
{{{\bf{s}}^T}({t_{1,n}})}& \cdots &{{{\bf{s}}^T}({t_{M,n}})}
\end{array}]^T}$, and the noise vector ${{\bf{n}}_{{\rm{r}}n}} =
{[\begin{array}{*{20}{c}}
{{\bf{n}}_{\rm{r}}^T({t_{1,n}})}& \cdots &{{\bf{n}}_{\rm{r}}^T({t_{M,n}})}
\end{array}]^T}$. Define the position parameters to be estimated using a vector
${\bf{p}} = {[{\bf{p}}_1^T,...,{\bf{p}}_K^T]^T}$. The CRB is a lower bound of the
variances of the parameters ${\bf{p}}$ for any unbiased estimator. Following the
steps of [23], and [5], we get the CRB of ${\bf{p}}$:

\begin{equation}
{\bf{CRB}}({\bf{p}}) = \frac{{\sigma _n^2}}{2}{\left[ {\sum\limits_{n = 1}^N
{{\mathop{\rm Re}\nolimits} \left( {{\bf{S}}_n^H{{\bf{D}}^H}{\bf{P}}_{\bf{A}}^
\bot {\bf{D}}{{\bf{S}}_n}} \right)} } \right]^{ - 1}}
\label{eq:54}
\end{equation}

{\raggedright
where
}

\[
{\bf{P}}_{\bf{A}}^ \bot  = {{\bf{I}}_{ML}} - {\bf{A}}{\left(
{{{\bf{A}}^H}{\bf{A}}} \right)^{ - 1}}{{\bf{A}}^H} \in {{\bf{C}}^{ML \times ML}},
\]

\[
{{\bf{S}}_n} = {{\bf{I}}_{3K}} \otimes {{\bf{s}}_n} \in {{\bf{C}}^{3M{K^2}
\times 3K}},
\]

\begin{equation}
{\bf{D}} = \left[ {\frac{{\partial {\bf{A}}}}{{\partial {x_1}}},\frac{{\partial
{\bf{A}}}}{{\partial {y_1}}},...,\frac{{\partial {\bf{A}}}}{{\partial
{x_P}}},\frac{{\partial {\bf{A}}}}{{\partial {y_P}}}} \right]
\label{eq:55}
\end{equation}

{\raggedright
The partial derivative of \textbf{A} with respect to ${x_p}$ is given by:
}

\begin{equation}
\frac{{\partial {\bf{A}}}}{{\partial {x_p}}} = \left[ {\begin{array}{*{20}{c}}
{\frac{{\partial {\bf{A}}({t_1})}}{{\partial {x_p}}}}&{}&{}\\
{}& \ddots &{}\\
{}&{}&{\frac{{\partial {\bf{A}}({t_M})}}{{\partial {x_p}}}}
\end{array}} \right]
\label{eq:56}
\end{equation}
and the partial derivative of ${\bf{A}}({t_m})$ with respect to ${x_p}$ is given
by:

\[
\frac{{\partial {\bf{A}}({t_m})}}{{\partial {x_p}}} =
[{\bf{0}},...,\frac{{\partial {{\bf{\alpha }}_p}({t_m})}}{{\partial
{x_p}}},...,{\bf{0}}],
\]
where

\begin{equation}
\frac{{\partial {{\bf{\alpha }}_p}({t_m})}}{{\partial {x_p}}} = {\left[ {\beta
_{11}^{(m)}(p),...,\beta _{L1}^{(m)}(p)} \right]^T}
\label{eq:57}
\end{equation}
where $\beta _{l1}^{(m)}(p)$ is defined as in \eqref{eq:51}. The partial
derivatives of \textbf{A} with respect to ${y_p}$ and ${z_p}$ have similar
formulas as \eqref{eq:56}.


\bibliographystyle{elsarticle-harv}


\end{document}